\documentclass[journal,comsoc]{IEEEtran}
\IEEEoverridecommandlockouts

\usepackage{graphicx}
\usepackage{amssymb}
\usepackage{amsmath}
\usepackage{caption}
\usepackage{cite}
\usepackage{mathrsfs}
\usepackage[displaymath,mathlines]{lineno}
\usepackage{color}
\usepackage{overpic}
\usepackage{multirow}
\usepackage{algpseudocode}
\usepackage{algorithm,algpseudocode}
\usepackage{algorithmicx}
\usepackage{pbox}
\usepackage{multicol}
\usepackage{lipsum}
\usepackage{amsthm}
\usepackage{lipsum}
\usepackage{makecell}
\usepackage{diagbox}
\usepackage{enumitem}

\makeatletter
\newcommand{\doublewidetilde}[1]{{%
		\mathpalette\double@widetilde{#1}%
}}
\newcommand{\double@widetilde}[2]{%
	\sbox\z@{$\m@th#1\widetilde{#2}$}%
	\ht\z@=.5\ht\z@
	\widetilde{\box\z@}%
}
\makeatother

\newtheorem{theorem}{Theorem}
\newtheorem{lemma}{Lemma}
\newtheorem{corollary}{Corollary}

\newtheorem{example}{Example}
\newtheorem{remark}{Remark}
\newtheorem{assumption}{Assumption}

\definecolor{red}{rgb}{0,0,0}
\begin{document}
%

\title{\huge Power Control in Cellular Massive MIMO with Varying User Activity: A Deep Learning Solution}

\author{\IEEEauthorblockN{\fontsize{11}{11}{Trinh Van Chien, Thuong Nguyen Canh,  \textit{Member}, \textit{IEEE}, Emil Bj\"{o}rnson, \textit{Senior Member}, \textit{IEEE},\\  and Erik G. Larsson, \textit{Fellow}, \textit{IEEE}}}
	\thanks{
		T. V. Chien was with the Department of Electrical
		Engineering (ISY), Link\"{o}ping University, 581 83 Link\"{o}ping, Sweden. He is now with the School of Electronics and Telecommunications, Hanoi University of Science and Technology, Vietnam (email: trinhchien.dt3@gmail.com). E. Bj\"{o}rnson and E. G. Larsson are with the Department of Electrical
		Engineering (ISY), Link\"{o}ping University, 581 83 Link\"{o}ping,
		Sweden (email: emil.bjornson@liu.se; erik.g.larsson@liu.se). T. N. Canh was with the Institute for Datability Science, Osaka University, Japan. He is now with Tencent America as a senior researcher (email: ngcthuong@ids.osaka-u.ac.jp). Parts of this paper were presented at ICC 2019 \cite{Chien2019a}.
	}
	\thanks{This paper was supported by the European Union's Horizon 2020 research and innovation programme under grant agreement No 641985 (5Gwireless). It was also supported by ELLIIT, CENIIT, and  the Vietnam’s Ministry of Education and Training (MOET) Project B2019-BKA-10.}
}

\maketitle

\begin{abstract}
This paper considers the sum spectral efficiency (SE) optimization problem in multi-cell Massive MIMO systems with a varying number of active users. This is formulated as a joint pilot and data power control problem. Since the problem is non-convex, we first derive a novel iterative algorithm that obtains a stationary point in polynomial time. To enable real-time implementation, we also develop a deep learning solution. The proposed neural network, PowerNet, only uses the large-scale fading information to predict both the pilot and data powers. The main novelty is that we exploit the problem structure to design a single neural network that can handle a dynamically varying number of active users; hence, PowerNet is simultaneously approximating many different power control functions with varying number inputs and outputs. This is not the case in prior works and thus makes PowerNet an important step towards a practically useful solution. Numerical results demonstrate that PowerNet only loses  $2\%$ in sum SE, compared to the iterative algorithm, in a nine-cell system with up to $90$ active users per in each coherence interval, and the runtime was only $0.03$ ms on a graphics processing unit (GPU). When good data labels are selected for the training phase, PowerNet can yield better sum SE than by solving the optimization problem with one initial point.
\end{abstract}

\begin{IEEEkeywords}
	Massive MIMO, Pilot and Data Power Control, Deep Learning, Convolutional Neural Network.
\end{IEEEkeywords}

\IEEEpeerreviewmaketitle

\section{Introduction}
Future networks must provide higher channel capacity, lower latency, and better quality of service than contemporary networks \cite{Andrews2014b}. These goals can only be achieved by drastic improvements of the wireless network architecture \cite{Wang2014}. Among potential candidates, Massive MIMO (multiple-input multiple-output) is an emerging physical layer technology which allows a base station (BS) equipped with many antennas to serve tens of users on the same time and frequency resource \cite{Marzetta2010a}. Utilizing the same spectrum and power budget, Massive MIMO can increase both spectral and energy efficiency by orders of magnitude compared with contemporary systems. This is because the propagation channels of different users decorrelate when increasing the number of antennas at each BS and strong array gains are achievable with little inter-user interference. 

Resource allocation is important in Massive MIMO networks to deal with the inter-user interference and, particularly, so-called pilot contamination \cite{Jose2011b, al2018successive}. Many resource allocation problems in Massive MIMO are easier to solve than in conventional systems since the channel hardening makes the utility functions only depend on the large-scale fading coefficients which are stable over a long time period \cite{Bjornson2016b}, while adaptation to the quickly varying small-scale fading is conventionally needed. 
Table~\ref{tableOptimizationProblems} categorizes the existing works on power control for cellular Massive MIMO in terms of utility functions and optimization variables. There are only a few works that jointly optimize the pilot and data powers, which is of key importance to deal with pilot contamination in multi-cell systems. In this paper, we optimize the sum SE with respect to the pilot and data powers. To the best of our knowledge, it is the first paper that considers this problem in
cellular Massive MIMO systems, where each BS serves a varying number of users.  
Note that we did not include single-cell papers in Table~\ref{tableOptimizationProblems}. For example, 
 the paper  \cite{Victor2017a} exploits the special structure arising from imperfect channel state information (CSI) in single-cell systems to maximize the sum SE using an efficient algorithm. This finds the globally optimal pilot and data powers, but it does not extend to multi-cell systems since the structure is entirely different.
 
Deep learning \cite{Goodfellow-et-al-2016} is a popular data-driven approach to solve complicated problems and has shown superior performance in various applications in image restoration, pattern recognition, etc. Despite its complicated and rather heuristic training phase, deep learning has recently shown promising results in communication applications \cite{o2017introduction, liang2018towards}. From the universal approximation theorem \cite{hornik1989}, deep learning can learn to approximate functions for which we have no closed-form expression. In particular, it can be used to approximate the solution to resource allocation problems and achieve greatly reduced computational complexity. The authors in \cite{sun2018learning} construct a fully-connected deep neural network to allocate transmit power in an effort to maximize the sum SE in a wireless system serving a few tens of users. The same network structure is reused in \cite{zappone2018model} to solve an energy-efficiency problem. Standard fully-connected feed-forward networks with many layers are used, but since the considered problems are challenging, the prediction performance is substantially lower than when directly solving the optimization problems, e.g., the loss varies from $5\%$ to $16\%$ depending on the system setting. By using the methodology of ensemble learning to select the best prediction results among multiple fully connected feed-forward neural networks, \cite{liang2018towards} is able to achieve better performance than when solving a sum SE optimization problem to a stationary point. The paper considers a multi-user single-input single-output system in which each BS serves a single user and has perfect channel state information (CSI). Moreover, previous neural network designs for resource allocation in wireless communications are utilizing the instantaneous CSI which is practically questionable, especially in cellular Massive MIMO systems. This is because the small-scale fading varies very quickly and the deep neural networks have very limited time to process the collection of all the instantaneous channel vectors, each having a number of parameters proportional to the number of BS antennas. The recent work in \cite{Luca2018DL} designs a neural network utilizing only statistical channel information to predict transmit powers in an equally-loaded cellular Massive MIMO system with spatially correlated fading. Although the prediction performance is good, the paper does not discuss how to generalize the approach to having varying number of users per cell.

\begin{table*}[t]
	\caption{Previous works on power control for cellular Massive MIMO}
	\centerline{ 
		\begin{tabular}{|c|c|c|}
			\hline
			\diagbox[width=11.5em]{Utility function }{Variables} & Data Powers Only & Joint Pilot \& Data \\
			\hline
			Minimum transmit power & \cite{senel2019joint, massivemimobook, Chien2016b} & \cite{Guo2014a} \\
			\hline 		
			Max-min fairness  & \cite{Marzetta2016a, adhikary2017uplink, Chien2016b, massivemimobook, Luca2018DL}  & \cite{Chien2017b,Ghazanfari2019, van2018joint}  \\
			\hline
			Maximum product SINR & \cite{massivemimobook, Luca2018DL} & \cite{dao2018disjoint} \\
			\hline
			Maximum sum SE & \cite{Chien2019ICCa, tran2016network, li2017massive, Chien2018a} & \textbf{This paper}\\
			\hline
		\end{tabular}
	} \vspace*{-0.2cm}
	\label{tableOptimizationProblems}
\end{table*}

The main issue with the previous deep learning solutions is that many different neural networks, with varying number of inputs and outputs, need to be trained and used depending on the number of active users. For example, if there are $L$ cells and between $0$ and $K_{\max}$ users per cell, then  \cite{Luca2018DL} would require $L2^{K_{\max}}$ different neural networks to cover all the cases that can appear. Even in the small setup of $L=4$ and $K_{\max}=5$ considered in \cite{Luca2018DL}, this requires $128$ different neural networks. We propose a solution to this problem by designing a single neural network that can handle a varying number of users; more precisely, we exploit the structure of the considered power control problem to train a neural network to simultaneously approximate $2^{K_{\max}}$ different power control functions, having varying number users (i.e., inputs/outputs).

In this paper, we consider the joint optimization of the pilot and data powers for maximum sum SE in multi-cell Massive MIMO systems. 
Our main contributions are:
\begin{itemize}
 \item We formulate a sum ergodic SE maximization problem, with the data and pilot powers as variables, where each cell may have a different number of active users. To overcome the inherent non-convexity, an equivalent problem with element-wise convex structure is derived. An alternating optimization algorithm is proposed to find a stationary point. Each iteration is solved in closed form.
 
 \item To reduce the computational complexity, design a deep convolutional neural network (CNN) that learns the solution to the alternating optimization algorithm, from one or multiple starting points. Our deep CNN is named PowerNet, has a residual structure, and is densely connected. PowerNet is designed for a maximum number of users and can manage any number of active users up to the maximum. The inputs to PowerNet are the large-scale fading coefficients between each active user and BS, while the outputs are the pilot and data powers. Zeros are used for inactive users. Note that the number of inputs/outputs is independent of the number of antennas.
 
 \item Numerical results manifest the effectiveness of the proposed alternating optimization algorithm as compared to the baseline of full transmit power. Meanwhile, PowerNet achieves highly accurate power prediction and a sub-milliseconds runtime.
\end{itemize}

The remainder of this paper is organized as follows: Section~\ref{Section: System Model} introduces our cellular Massive MIMO system model, with a varying number of users per cell, and the basic ergodic SE analysis. We formulate and solve the joint pilot and data power control problem for maximum sum SE in Section~\ref{Section:JointPilotDataforSumSE}. The proposed low complexity deep learning solution is given in Section~\ref{Section:CNN}. Finally, numerical results are shown in Section~\ref{Section:NumericalResults} and we provide the main conclusions in Section~\ref{Section:Conclusion}.

\textit{Notation}: Upper (lower) bold letters are used to denote matrices (vectors). $\mathbb{E} \{ \cdot \}$ is the expectation of a random variable. $(\cdot)^H$ is the Hermitian transpose and the cardinality of set $\mathcal{A}$ is $| \mathcal{A}|$. We let $\mathbf{I}_M$ denote the $M \times M$ identity matrix. $\mathbb{C}, \mathbb{R},$ and $\mathbb{R}_{+}$ denote the complex, real and non-negative real field, respectively. The floor operator denotes as $\left \lfloor \cdot \right \rfloor$ and the Frobenius norm as $\| . \|_F$. Finally, $\mathcal{CN}(\cdot, \cdot)$ is circularly symmetric complex Gaussian distribution.

\section{Dynamic Massive MIMO System Model} \label{Section: System Model}
We consider a multi-cell Massive MIMO system comprising of $L$ cells, each having a BS equipped with $M$ antennas. We call it a dynamic system model since each BS is able to serve $K_{\max}$ users, but maybe only a subset of the users are active at any given point in time. We will later model the active subset of users randomly and exploit this structure when training a neural network. Since the wireless channels vary over time and frequency, we consider the standard block fading model \cite{Marzetta2016a} where the time-frequency resources are divided into coherence intervals of $\tau_c$ modulation symbols for which the channels are static and frequency flat. At an arbitrary given coherence interval, BS~$l$ is serving a subset of active users. We define the set $\mathcal{A}_l $ containing the indices of all active users in cell~$l$, for which $0 \leq | \mathcal{A}_l | \leq K_{\max}$. The channel between active user~$t \in \mathcal{A}_i$ in cell~$i$ and BS~$l$ is denoted as $\mathbf{h}_{i,t}^l \in \mathbb{C}^{M}$ and follows an independent and identically distributed (i.i.d.) Rayleigh fading distribution:
\begin{equation}
\mathbf{h}_{i,t}^l \sim \mathcal{CN} \left(\mathbf{0}, \beta_{i,t}^l \mathbf{I}_{M} \right),
\end{equation}
where $\beta_{i,t}^l\geq0$ is the large-scale fading coefficient that models geometric pathloss and shadow fading. In this paper, we consider uncorrelated Rayleigh fading because the spectral efficiency obtained with this tractable model well matches the results obtained in non-line-of-sight measurement \cite{Gao2015a}. However, the data-driven part presented later in this paper can applied together with any other channel model. The distributions are known at the BSs, but the realizations are unknown and need to be estimated in every coherence interval using a pilot transmission phase.
\subsection{Uplink Pilot Transmission Phase}
We assume that a set of $K_{\max}$ orthonormal pilot signals are used in the system. 
User~$k$ in each cell is preassigned the pilot $\pmb{\psi}_{k} \in \mathbb{C}^{K_{\max}}$ with $\|\pmb{\psi}_{k} \|^2 = K_{\max}$, no matter if the user is active or not in the given coherence interval, but this pilot is only transmitted when the user has data to transmit (or receive). This pilot assignment guarantees that there is no intra-cell pilot contamination. The channel estimation of a user is interfered by the other users that use the pilot signal, which is called pilot contamination. The received baseband pilot signal $\mathbf{Y}_l \in \mathbb{C}^{M \times K_{\max} }$ at BS~$l$ is
\begin{equation}
\mathbf{Y}_l = \sum_{i=1 }^L \sum_{t \in \mathcal{A}_i } \sqrt{\hat{p}_{i,t}}\mathbf{h}_{i,t}^l \pmb{\psi}_{t}^H + \mathbf{N}_l,
\end{equation}
where $\mathbf{N}_l \in \mathbb{C}^{M \times K }$ is the additive noise with  i.i.d.~$\mathcal{CN}(0, \sigma_{\mathrm{UL}}^2)$ elements. Meanwhile, $\hat{p}_{i,t}$ is the pilot power that active user~$t$ in cell~$i$ allocates to its pilot transmission. The channel between a particular user~$t  \in \mathcal{A}_i$ in cell~$i$ and BS~$l$ is estimated from
\begin{equation}
\mathbf{y}_{i,t} = \mathbf{Y}_l \pmb{\psi}_{t} = \sum_{i' \in \mathcal{P}_t } \sqrt{\hat{p}_{i',t}} \mathbf{h}_{i',t}^l \pmb{\psi}_{t}^H  \pmb{\psi}_{t}+ \mathbf{N}_l \pmb{\psi}_{t},
\end{equation}
where the set $\mathcal{P}_t$ contains the indices of cells having user~$t$ in active mode, which is formulated from the user activity set of each cell  as
\begin{equation}
\mathcal{P}_t = \left\{ i' \in \{1, \ldots,L \}:  t \in \mathcal{A}_{i'} \right\}.
\end{equation} 
By using minimum mean square error  (MMSE) estimation~\cite{Kay1993a}, the channel estimate of an arbitrary active user is as follows. 
\begin{lemma}\label{Lemma:DynamicChannelEst}
	If BS~$l$ uses MMSE estimation, the channel estimate of active user~$t$ in cell~$i$ is
	\begin{equation}
	\hat{\mathbf{h}}_{i,t}^l = \mathbb{E} \left\{ \mathbf{h}_{i,t}^l | \mathbf{y}_{i,t} \right\} = \frac{\beta_{i,t}^l \sqrt{\hat{p}_{i,t}} }{K_{\max} \sum_{i' \in \mathcal{P}_t } \hat{p}_{i',t} \beta_{i',t}^l + \sigma_{\mathrm{UL}}^2  } \mathbf{y}_{i,t},
	\end{equation}
	which follows a complex Gaussian distribution as
	\begin{equation} \label{eq:ChannelEstDynamicMRC}
	\begin{split}
	\hat{\mathbf{h}}_{i,t}^l &\sim \mathcal{CN} \left(\mathbf{0}, \frac{K_{\max} (\beta_{i,t}^l)^2 \hat{p}_{i,t}  }{K_{\max} \sum_{i' \in \mathcal{P}_t } \hat{p}_{i',t} \beta_{i',t}^l + \sigma_{\mathrm{UL}}^2  } \mathbf{I}_{M} \right).
	\end{split}
	\end{equation}
	The estimation error  $\mathbf{e}_{i,t}^l =  \mathbf{h}_{i,t}^l - \hat{\mathbf{h}}_{i,t}^l$ is independently distributed as
	\begin{equation} \label{eq:estimationerror}
	\mathbf{e}_{i,t}^l \sim \mathcal{CN} \left(\mathbf{0},  \frac{K_{\max} \sum_{i' \in \mathcal{P}_t \setminus \{ i\} } \hat{p}_{i',t} \beta_{i,t}^l \beta_{i',t}^l + \beta_{i,t}^l \sigma_{\mathrm{UL}}^2 }{K_{\max} \sum_{i' \in \mathcal{P}_t } \hat{p}_{i',t} \beta_{i',t}^l + \sigma_{\mathrm{UL}}^2  } \mathbf{I}_M \right).
	\end{equation}
\end{lemma}
\begin{proof}
The proof follows directly from standard MMSE estimation techniques \cite{Kay1993a, massivemimobook}. 
\end{proof}
The statistical information in Lemma~\ref{Lemma:DynamicChannelEst} of each channel estimate and estimation error are used to construct the linear combining vectors and to derive a closed-form expression of the uplink~SE. 

\vspace*{-0.25cm}
\subsection{Uplink Data Transmission Phase}
During the uplink data transmission phase, every active user~$t \in \mathcal{A}_i$ in cell~$i$ transmits a data symbol $s_{i,t}$ with $\mathbb{E} \{ |s_{i,t}|^2\} =1$. The received signal $\mathbf{y}_{l} \in \mathbb{C}^{M}$ at BS~$l$ is the superposition of signals from all users across cells:
\begin{equation} \label{eq:ReceivedData}
\mathbf{y}_{l} = \sum_{i=1}^L \sum_{t \in \mathcal{A}_i } \sqrt{p_{i,t}} \mathbf{h}_{i,t}^l s_{i,t} + \mathbf{n}_{l},
\end{equation}
where $p_{i,t}$ is the power that active user $t$ in cell $i$ allocates to the data symbol $s_{i,t}$ and $\mathbf{n}_{l} \in \mathbb{C}^{M}$ is complex Gaussian noise distributed as $\mathcal{CN} \left(\mathbf{0}, \sigma_{\mathrm{UL}}^2 \mathbf{I}_{M} \right)$. Each BS uses maximum ratio combining (MRC) to
detect the desired signals from its users. In particular, BS~$l$ selects the combining vector for its user~$k$ as
$\mathbf{v}_{l,k} = \hat{\mathbf{h}}_{l,k}^l,$
and we will quantify the achievable spectral efficiency by using the use-and-then-forget capacity bounding technique \cite{Marzetta2016a}. The closed-form expression of the lower bound on the uplink capacity is shown in Lemma~\ref{LemmaMRC}.
\begin{lemma} \label{LemmaMRC}
If each BS uses MRC for data detection, a closed-form expression for the uplink ergodic SE of active user~$k$ in cell~$l$ is 
\begin{equation} \label{eq:ULRateDynamicMRC}
R_{l,k} \left( \{ \hat{p}_{i,t},p_{i,t} \}\right) = \left(1 - \frac{K_{\max}}{\tau_c} \right) \log_2 \left( 1 + \mathrm{SINR}_{l,k} \right),
\end{equation}
where the effective SINR value of this user is
\begin{equation} \label{eq:DynamicSINRlk}
\mathrm{SINR}_{l,k} = M K_{\max} p_{l,k}  \hat{p}_{l,k} (\beta_{l,k}^l)^2 / \mathit{D}_{l,k}
\end{equation} 
and
\begin{equation} \label{eq:Denominatorlk}
\begin{split}
&\mathit{D}_{l,k} = M K_{\max}  \sum\limits_{i \in \mathcal{P}_k \setminus \{l\}} p_{i,k} \hat{p}_{i,k} (\beta_{i,k}^l)^2  + \\
& \left( K_{\max}  \sum\limits_{i \in \mathcal{P}_k } \hat{p}_{i,k} \beta_{i,k}^l + \sigma_{\mathrm{UL}}^2  \right)\left( \sum\limits_{i=1}^L \sum\limits_{t \in \mathcal{A}_i } p_{i,t} \beta_{i,t}^l + \sigma_{\mathrm{UL}}^2 \right).
\end{split}
\end{equation}
\end{lemma}
\begin{proof}
The proof follows along the lines of Corollary~$4.5$ in \cite{massivemimobook} except for the different notation and the fact that each user can assign a different power to its pilot and data.
\end{proof}
The numerator of the SINR expression in \eqref{eq:DynamicSINRlk} indicates contributions of the array gain which is directly proportional to the number of antennas at the serving BS. The first part in the denominator represents the pilot contamination effect and it is also proportional to the number of BS antennas. Interestingly, active user~$k$ in cell~$l$ will have unbounded capacity when $M \rightarrow \infty$ if all users using the same pilot sequence $\pmb{\psi}_k$ are silent (i.e., inactive or allocated zero transmit power), \textcolor{red}{thanks to the massive antenna array gain while inter-cell mutual interference and noise are negligible}. We notice that the remaining terms are non-coherent mutual interference and noise that can have a vanishing impact when the number of antennas grow. Furthermore, the SE of a user is proportional to $(1 - K_{\max}/\tau_c)$, which is the pre-log factor in \eqref{eq:ULRateDynamicMRC}. This is the fraction of symbols per coherence interval that are used for data transmission, which thus reduces when the number of pilots is increased. In the special case of $ |\mathcal{A}_1 | = \ldots =  |\mathcal{A}_L |$, the analytical results in Lemma~\ref{LemmaMRC} particularize to equally-loaded systems as in the previous works. That special case is unlikely to occur in practice since the data traffic is generated independently for each user.
\section{Joint Pilot and Data Power Control for Sum Spectral Efficiency Optimization} \label{Section:JointPilotDataforSumSE}
The main goal of this paper is to solve the sum SE maximization problem. Achieving high SE is important for future networks, and the (weighted) sum SE maximization is also the core problem to be solved in practical algorithms for dynamic resource allocation \cite{Georgiadis2006a}. The previous works \cite{Victor2017a,Chien2018a} consider this problem for single-cell systems with joint pilot and data power control or multi-cell systems with only data power control, respectively. In contrast, we formulate and solve a sum SE maximization problem with joint pilot and data power control. This optimization problem has not been tackled before in the Massive MIMO literature due to its inherent non-convexity structure. In this section, we develop an iterative algorithm that achieves a stationary point in polynomial time by solving a series of convex sub-problems in closed form.

\subsection{Problem Formulation}
The considered optimization problem is to maximize the sum SE of all active users under power constraints on each transmitted symbol:
\begin{equation} \label{Prob:SumRate}
\begin{aligned}
& \underset{\{ \hat{p}_{l,k}, p_{l,k} \} }{\textrm{maximize}}
&&   \sum_{l=1}^{L} \sum_{ k \in \mathcal{A}_l } R_{l,k} \left( \{ \hat{p}_{i,t},p_{i,t} \}\right) \\
& \textrm{subject to}
&&   0 \leq \hat{p}_{l,k} \leq P_{l,k}, \; \forall l,k,\\
&&& 0 \leq p_{l,k} \leq  P_{l,k}, \; \forall l,k,
\end{aligned}
\end{equation}
where $P_{l,k} \geq 0$ is the maximum power that user~$k$ in cell~$l$ can supply to each transmitted symbol. Problem~\eqref{Prob:SumRate} is independent of the small-scale fading, so it allows for long-term performance optimization, if the users are continuously  active and there is no large-scale user mobility. However, in practical systems, some users are moving quickly 
and new scheduling decisions are made every few milliseconds based on the users' traffic. It is therefore important to be able to solve \eqref{Prob:SumRate} very quickly to adapt to these changes.\footnote{Note that the ergodic SE is a reasonable performance metric also in this scenario, since long codewords can span over the frequency domain and the channel hardening makes the channel after MRC almost deterministic. The simulations in \cite{Bjornson2016b} shows that coding over 1\,kB of data is sufficient to operate closely to the ergodic SE.} The sum SE optimization problem is non-convex in general and seeking the optimal solution has very high complexity in any non-trivial setup \cite{Al-Shatri2012a}. However, the pilot and data power constraints in \eqref{Prob:SumRate} guarantee a convex feasible set and make all ergodic SE expressions continuous functions of the power variables. According to Weierstrass’ theorem \cite{Horn2013a}, an optimal solution set of pilot and data powers always exist. The global optimum to problem~\eqref{Prob:SumRate} can be obtained by using the branch-and-bound approach, similar to \cite{joshi2011weighted} for conventional multiple-input single-output channels. However, the complexity of that approach grows exponentially with the number of users, thus, even if off-line optimization is considered, it can only be utilized in small-scale networks with a few users. Moreover, in the settings where it can be utilized, the stationary points obtained by the weighted MMSE methodology was shown in \cite{joshi2011weighted, Christensen2008a} to be close to the global optimum.

Inspired by the weighted MMSE methodology \cite{joshi2011weighted, Christensen2008a}, we will now propose an iterative algorithm to find a stationary point to \eqref{Prob:SumRate}. By removing the pre-log factor and setting $\hat{\rho}_{l,k} = \sqrt{\hat{p}_{l,k} }$ and $\rho_{l,k} = \sqrt{p_{l,k} }$, $\forall l,k,$ as the new optimization variables, we formulate a new problem that is equivalent with \eqref{Prob:SumRate}.
\begin{theorem} \label{Theorem:WMMSEMRC}
The following optimization problem is equivalent to problem \eqref{Prob:SumRate}:
\begin{equation} \label{Prob:WMMSEv1}
\begin{aligned}
& \underset{\substack{ \{ w_{l,k} \geq 0, u_{l,k} \}, \\ \{ \hat{\rho}_{l,k}, \rho_{l,k} \geq 0 \} }}{\mathrm{minimize}}
&&   \sum_{l=1}^{L} \sum_{k \in \mathcal{A}_l } w_{l,k} e_{l,k} - \ln (w_{l,k}) \\
& \,\,\mathrm{subject\,\,to}
&&   \hat{\rho}_{l,k}^2 \leq P_{l,k}, \; \forall l,k,\\
&&& \rho_{l,k}^2 \leq  P_{l,k}, \; \forall l,k,
\end{aligned}
\end{equation}
where 
\begin{equation} \label{Prob:WMMSEv2}
\begin{split}
e_{l,k} =& MK_{\max} u_{l,k}^2 \sum_{i \in \mathcal{P}_k} \rho_{i,k}^2 \hat{\rho}_{i,k}^2 (\beta_{i,k}^l)^2 - 2 \sqrt{M  K_{\max}} \rho_{l,k} \hat{\rho}_{l,k} \times \\
&u_{l,k} \beta_{l,k}^l  + u_{l,k}^2 \left( K_{\max} \sum_{i\in \mathcal{P}_k} \hat{\rho}_{i,k}^2  \beta_{i,k}^l  + \sigma_{\mathrm{UL}}^2 \right)\times \\
& \left( \sum_{i=1}^L \sum_{t \in \mathcal{A}_i } \rho_{i,t}^2 \beta_{i,t}^l + \sigma_{\mathrm{UL}}^2 \right) +1,
\end{split}
\end{equation}
in the sense that if $\{u_{l,k}^{\ast}, w_{l,k}^{\ast}, \hat{\rho}_{l,k}^{\ast}, \rho_{l,k}^{\ast} \}$ is a global optimum to problem \eqref{Prob:WMMSEv1}, then  $\{ (\hat{\rho}_{l,k}^{\ast})^2, (\rho_{l,k}^{\ast})^2 \}$ is a global optimum to problem \eqref{Prob:SumRate}.
\end{theorem}
\begin{proof}
The proof consists of two main steps: the mean square error $e_{l,k}$ is first formulated by considering a single-input single-output (SISO) communication system with deterministic channels having the same SE as in Lemma~\ref{LemmaMRC}, where $u_{l,k}$ is the beamforming coefficient utilized in such a SISO system and $w_{l,k}$ is the weight value in the receiver. After that, the equivalence of two problems \eqref{Prob:SumRate} and \eqref{Prob:WMMSEv1} is obtained by finding the optimal solution of $u_{l,k}$ and $w_{l,k},\forall l,k,$ given the other optimization variables. The detailed proof is given in Appendix~\ref{Appendix:WMMSEMRC}.
\end{proof}
The new problem formulation in Theorem~\ref{Theorem:WMMSEMRC} is still non-convex, but it has an important desired property: if we consider one of the sets $\{ u_{l,k}\}$, $\{ w_{l,k}\}$, $\{ \hat{\rho}_{l,k}\}$, and $\{ \rho_{l,k} \}$ as the only optimization variables, while the other variables are constant, then problem \eqref{Prob:WMMSEv2} is convex. Note that the set of optimization variables and SE expressions are different than in the previous works \cite{Shi2011, Weeraddana2012a} that followed similar paths of reformulating their sum SE problems, which is why Theorem~\ref{Theorem:WMMSEMRC} is a main contribution of this paper. From the closed-form SE expression, both channel estimation errors and pilot contamination effects are considered. In particular, we are optimizing both pilot and data powers, which is not the case in prior work. Moreover, in our case we can get closed-form solutions in each iteration, leading to a particularly efficient implementation. We exploit this property to derive an iterative algorithm to find a local optimum (stationary point) to \eqref{Prob:WMMSEv2} as shown in the following subsection.
\subsection{Iterative Algorithm}
	\begin{algorithm}[t]
	\caption{Alternating optimization approach for \eqref{Prob:WMMSEv1}} \label{Algorithm:AlternatingApproach}
	\textbf{Input}:  Large-scale fading $\beta_{i,t}^{l}, \forall, i,t,l$; Maximum power levels $P_{l,k}, \forall l,k$; Initial values $\hat{\rho}_{l,k}^{(0)}$ and $\rho_{l,k}^{(0)}, \forall l,k$.  Set up $n=0$. \\
	\textbf{While} \textit{Stopping criterion \eqref{eq:Stoping} is not satisfied} \textbf{do}
		\begin{itemize}
			\item[1.] Set $n = n+1$.
			\item[2.] Update $u_{l,k}^{(n)}$, for all $l,k,$ by using \eqref{eq:ulkDynamicMRC} where every $\tilde{u}_{l,k}^{(n-1)}$ is computed by \eqref{eq:tildeulkDynamicMRC}.
			\item[3.] Update $w_{l,k}^{(n)}$, for all $l,k,$ by using \eqref{eq:wlkDynamicMRC} where every $e_{l,k}^{(n)}$ is computed by \eqref{eq:elknDynamicMRC}.
			\item[4.] Update $\hat{\rho}_{l,k}^{(n)}$, for all $l,k,$ by using \eqref{eq:dynamicrhohatlk}.
			\item[5.] Update $\rho_{l,k}^{(n)}$, for all $l,k,$ by using \eqref{eq:dynamicrholk}.
		    \item[6.] Store the currently solution: $\hat{\rho}_{l,k}^{(n)}$ and $\rho_{l,k}^{(n)}$, $\forall l,k$. 
		\end{itemize}
	 \textbf{End While} \\
	\textbf{Output}: The stationary point: $\hat{\rho}_{l,k}^{\mathrm{opt}} = \rho_{l,k}^{(n)}, \rho_{l,k}^{\mathrm{opt}}= \rho_{l,k}^{(n)} \forall l,k.$
\end{algorithm}
This subsection provides an iterative algorithm to obtain a stationary point to problem \eqref{Prob:WMMSEv1} by alternating between updating the different sets of optimization variables. This procedure is established by the following theorem.

\begin{theorem} \label{Theorem:IterativeAl}
	From an initial point $\{ \hat{\rho}_{l,k}^{(0)}, \rho_{l,k}^{(0)} \}$ satisfying the constraints, a stationary point to problem~\eqref{Prob:WMMSEv1} is obtained by updating $\{ u_{l,k}, w_{l,k}, \hat{\rho}_{l,k}, \rho_{l,k} \}$ in an iterative manner. At iteration~$n$, the variables 
	 are updated as follows:
\begin{itemize}[leftmargin=*]
	\item The $u_{l,k}$ variables, for all $l,k \in \mathcal{A}_l,$ are updated as
	\begin{equation}  \label{eq:ulkDynamicMRC}
	u_{l,k}^{(n)} = \sqrt{M  K_{\max}} \rho_{l,k}^{(n-1)} \hat{\rho}_{l,k}^{(n-1)} \beta_{l,k}^l / \tilde{u}_{l,k}^{(n-1)},
	\end{equation}
	where 
	\begin{equation} \label{eq:tildeulkDynamicMRC}
	\begin{split}
	&\tilde{u}_{l,k}^{(n-1)} = M K_{\max} \sum\limits_{i \in \mathcal{P}_k} (\rho_{i,k}^{(n-1)})^2 (\hat{\rho}_{i,k}^{(n-1)})^2  (\beta_{i,k}^l)^2 +\\
	&
	\left( K_{\max} \sum\limits_{i \in \mathcal{P}_k} (\hat{\rho}_{i,k}^{(n-1)})^2   \beta_{i,k}^l + \sigma_{\mathrm{UL}}^2  \right)\left( \sum\limits_{i=1}^L \sum_{t \in \mathcal{A}_i } (\rho_{i,t}^{(n-1)})^2 \beta_{i,t}^l + \sigma_{\mathrm{UL}}^2 \right).
	\end{split}
	\end{equation}
	\item The variables $w_{l,k}$, for all $l,k \in \mathcal{A}_l,$ are updated as
	\begin{equation} \label{eq:wlkDynamicMRC}
	w_{l,k}^{(n)} = 1/e_{l,k}^{(n)},
	\end{equation}
	where 
	\begin{equation} \label{eq:elknDynamicMRC}
	e_{l,k}^{(n)} = (u_{l,k}^{(n)})^2 \tilde{u}_{l,k}^{(n-1)}  - 2 \sqrt{M  K_{\max}} \rho_{l,k}^{(n-1)} \hat{\rho}_{l,k}^{(n-1)} u_{l,k}^{(n)} \beta_{l,k}^l +1.
	\end{equation}
	\item The variables $\hat{\rho}_{l,k}$, for all $l,k \in \mathcal{A}_l,$  are updated as in \eqref{eq:dynamicrhohatlk}:
	\begin{equation} \label{eq:dynamicrhohatlk}
	\hat{\rho}_{l,k}^{(n)} = \min \left( \frac{ \sqrt{M  K_{\max}} \rho_{l,k}^{(n-1)} u_{l,k}^{(n)} w_{l,k}^{(n)} \beta_{l,k}^l}{ \hat{\eta}_{l,k}^{(n)}} , \sqrt{P_{l,k}}\right),
	\end{equation} 
	where $\hat{\eta}_{l,k}^{(n)}$ is given by
	\begin{equation}
	\begin{split}
	\hat{\eta}_{l,k}^{(n)} =& (\rho_{l,k}^{(n-1)})^2 M K_{\max} \sum\limits_{i \in \mathcal{P}_k} w_{i,k}^{(n)} (u_{i,k}^{(n)})^2 (\beta_{l,k}^i)^2  + \\
	& K_{\max} \sum\limits_{j \in \mathcal{P}_k} w_{j,k}^{(n)}  (u_{j,k}^{(n)})^2 \beta_{l,k}^j \left( \sum\limits_{i=1}^L \sum\limits_{t \in \mathcal{A}_i } (\rho_{i,t}^{(n)})^2 \beta_{i,t}^j + \sigma_{\mathrm{UL}}^2 \right).
	\end{split}
	\end{equation}
	\item The variables $\rho_{l,k}$, for all $l,k \in \mathcal{A}_l,$ are updated as in \eqref{eq:dynamicrholk}:
	\begin{equation} \label{eq:dynamicrholk}
	\begin{split}
	\rho_{l,k}^{(n)} = \min \left( \frac{\sqrt{MK_{\max}} \hat{\rho}_{l,k}^{(n)} u_{l,k}^{(n)} w_{l,k}^{(n)} \beta_{l,k}^l }{ \eta_{l,k}^{(n)}  }, \sqrt{P_{l,k}}\right),
	\end{split}
	\end{equation} 
	where $\eta_{l,k}^{(n)}$ is given by
	\begin{equation}
	\begin{split}
	\eta_{l,k}^{(n)} =& (\hat{\rho}_{l,k}^{(n)})^2 M K_{\max}  \sum\limits_{i \in \mathcal{P}_k} w_{i,k}^{(n)} (u_{i,k})^2 (\beta_{l,k}^i)^2 +\\
	& \sum\limits_{i=1}^L \sum\limits_{t \in \mathcal{A}_i} w_{i,t}^{(n)} (u_{i,t}^{(n)})^2 \beta_{l,k}^i \left( K_{\max} \sum\limits_{j \in \mathcal{P}_k } (\hat{\rho}_{j,t}^{(n)})^2 \beta_{j,t}^i + \sigma_{\mathrm{UL}}^2\right).
	\end{split}
	\end{equation}
\end{itemize}

This iterative process converges to a stationary point $\{ u_{l,k}^{\mathrm{opt}}, w_{l,k}^{\mathrm{opt}}, \hat{\rho}_{l,k}^{\mathrm{opt}}, \rho_{l,k}^{\mathrm{opt}}\}$ to problem \eqref{Prob:WMMSEv1} and  then $\{ (\hat{\rho}_{l,k}^{\mathrm{opt}})^2, (\rho_{l,k}^{\mathrm{opt}})^2 \}$ is also a stationary point to problem \eqref{Prob:SumRate}. 
\end{theorem}
\begin{proof}
The proof derives the closed-form optimal solutions in \eqref{eq:ulkDynamicMRC}--\eqref{eq:dynamicrholk} to each of the optimization variables, when the other are fixed, by taking the first derivative of the partial Lagrangian function of \eqref{Prob:WMMSEv1} and equating it to zero. The fact that problems~\eqref{Prob:SumRate} and \eqref{Prob:WMMSEv1} have the same set of stationary points is further confirmed by the chain rule. The proof is given in Appendix~\ref{Appendix:IterativeAlgMRC}.
\end{proof}
Theorem~\ref{Theorem:IterativeAl} provides an iterative algorithm that obtains a local optimum to \eqref{Prob:SumRate} and \eqref{Prob:WMMSEv1} with relatively low computational complexity because of the closed-form solutions in each iteration. Algorithm~\ref{Algorithm:AlternatingApproach} gives a summary of this iterative process. From any feasible initial set of powers $\{ \hat{\rho}_{l,k}^{(0)}, \rho_{l,k}^{(0)} \}$, in each iteration, we update each optimization variable according to \eqref{eq:ulkDynamicMRC}--\eqref{eq:dynamicrholk}. This iterative process will be terminated when the variation of two consecutive iterations is small. For instance the stopping condition may be defined for a given accuracy $\epsilon > 0$ as
\begin{equation} \label{eq:Stoping}
\left| \sum_{l=1}^L \sum_{k \in \mathcal{A}_l } R_{l,k}^{(n)} - \sum_{l=1}^L \sum_{k \in \mathcal{A}_l } R_{l,k}^{(n-1)} \right| \leq \epsilon.
\end{equation}

By considering the multiplications, divisions, and logarithms as the dominated complexity, the number of arithmetic operations need for Algorithm~\ref{Algorithm:AlternatingApproach} to reach $\epsilon$-accuracy is
\begin{equation}
N_1 \left(8 \sum_{i=1}^L | \mathcal{A}_i| + 4 |\mathcal{P}_k | \sum_{i=1}^L | \mathcal{A}_i|  + 23 |\mathcal{P}_k| + 50 \right) \sum_{i=1}^L | \mathcal{A}_i| ,
\end{equation}
where $N_1$ is the number iterations required for the convergence which depends on the given $\epsilon-$accuracy. From Theorem~\ref{Theorem:IterativeAl}, we further observe the following result which is key to achieve an efficient deep learning solution that can handle varying number of users.

\begin{corollary} \label{corollary:phatlkandplk}
If a user has the large-scale fading coefficient equal to zero, then it will always get zero transmit powers when using the algorithm in Theorem~\ref{Theorem:IterativeAl}. Hence, an equivalent way of managing inactive users is to set their large-scale fading coefficients to zero and use $\mathcal{A}_l=\{ 1, \ldots, K_{\max}\}$.
\end{corollary}
\begin{proof}
\textcolor{red}{The numerators of \eqref{eq:dynamicrhohatlk} and \eqref{eq:dynamicrholk} directly imply that any user with zero large-scale fading ($\beta_{l,k}^l =0$) will be allocated zero pilot power ($\hat{p}_{l,k} =0$) and zero data power ($p_{l,k} =0$).}
\end{proof}

 \textcolor{red}{Note that, notwithstanding Corollary~\ref{corollary:phatlkandplk}, the system may reject some active users that have small but non-zero large-scale fading coefficients since Algorithm~\ref{Algorithm:AlternatingApproach} can assign zero power to these ones---similar to the behavior of standard waterfilling algorithms.} This is a key benefit of sum SE maximization as compared to max-min fairness power control \cite{Marzetta2016a, adhikary2017uplink, Chien2016b, massivemimobook, Luca2018DL,Chien2017b,Ghazanfari2019} 
and maximum product-SINR power control \cite{massivemimobook,Luca2018DL,dao2018disjoint}, which always allocate non-zero power to all users and, therefore, require an additional heuristic user admission control step for selecting which users to drop from service due to their poor channel conditions. If a particular user~$t$ in cell~$i$ is not served this implies that $\hat{p}_{i,t}^{\mathrm{opt}} = 0$ and $p_{i,t}^{\mathrm{opt}} = 0$. Hence, this user is neither transmitting in the pilot nor data phase. Corollary~\ref{corollary:phatlkandplk} will enable us to design a single neural network that can mimic Algorithm~\ref{Algorithm:AlternatingApproach} for any number of active users.

\begin{figure}[t]
	\centering
	\includegraphics[trim=0cm 0cm 0cm 0cm, clip=true, width=3.0in]{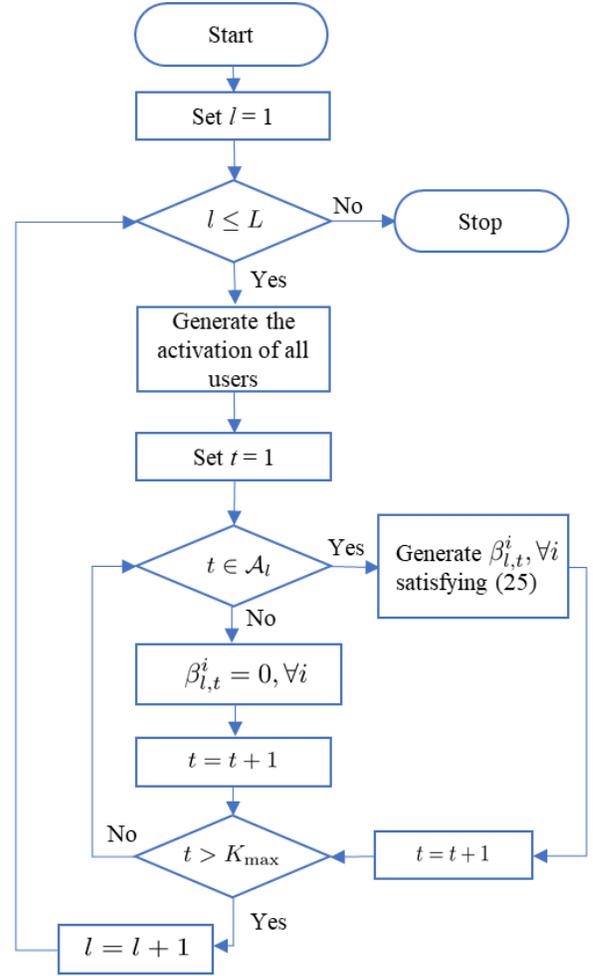} \vspace*{-0.0cm}
	\caption{The flowchart of generating one realization of the Massive MIMO network with $LK_{\max}$ users having random large-scale fading realizations and activity.}
	\label{FigLFSDig}
	\vspace*{-0.2cm}
\end{figure}
\begin{figure}[t]
	\centering
	\begin{overpic}[trim=1.0cm 0.5cm 1.2cm 0.8cm, clip=true, width=3.2in]{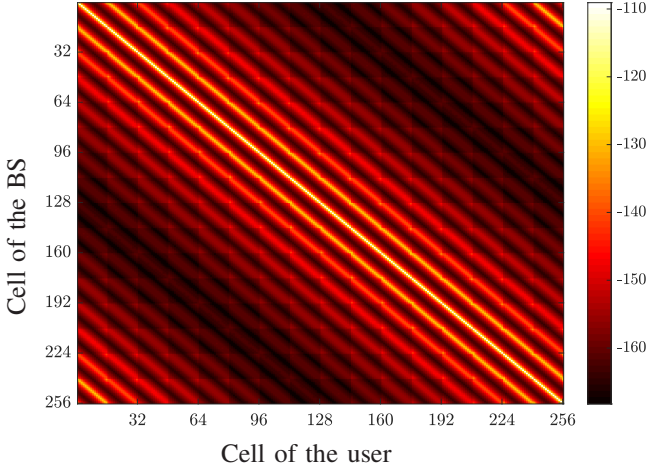}  
		\put (30,-4) {\fontsize{7}{7}{Cell of the user}}
		\put(-5,20){\fontsize{7}{7}{\rotatebox{90}{Cell of the BS }}}
	\end{overpic}  \vspace*{0.3cm}
	\caption{The pattern created by the average large-scale fading coefficients in dB between a user in a given cell and the BS in another cell, a coverage area $25$ km$^2$ with $L=256$ BSs on a $16 \times 16$ grid.}
	\label{FigPatter}
	\vspace*{-0.2cm}
\end{figure}
\begin{figure*}[t]
	\centering
	\includegraphics[trim=0.5cm 6.5cm 22cm 20.5cm, clip=true, width=5.6in]{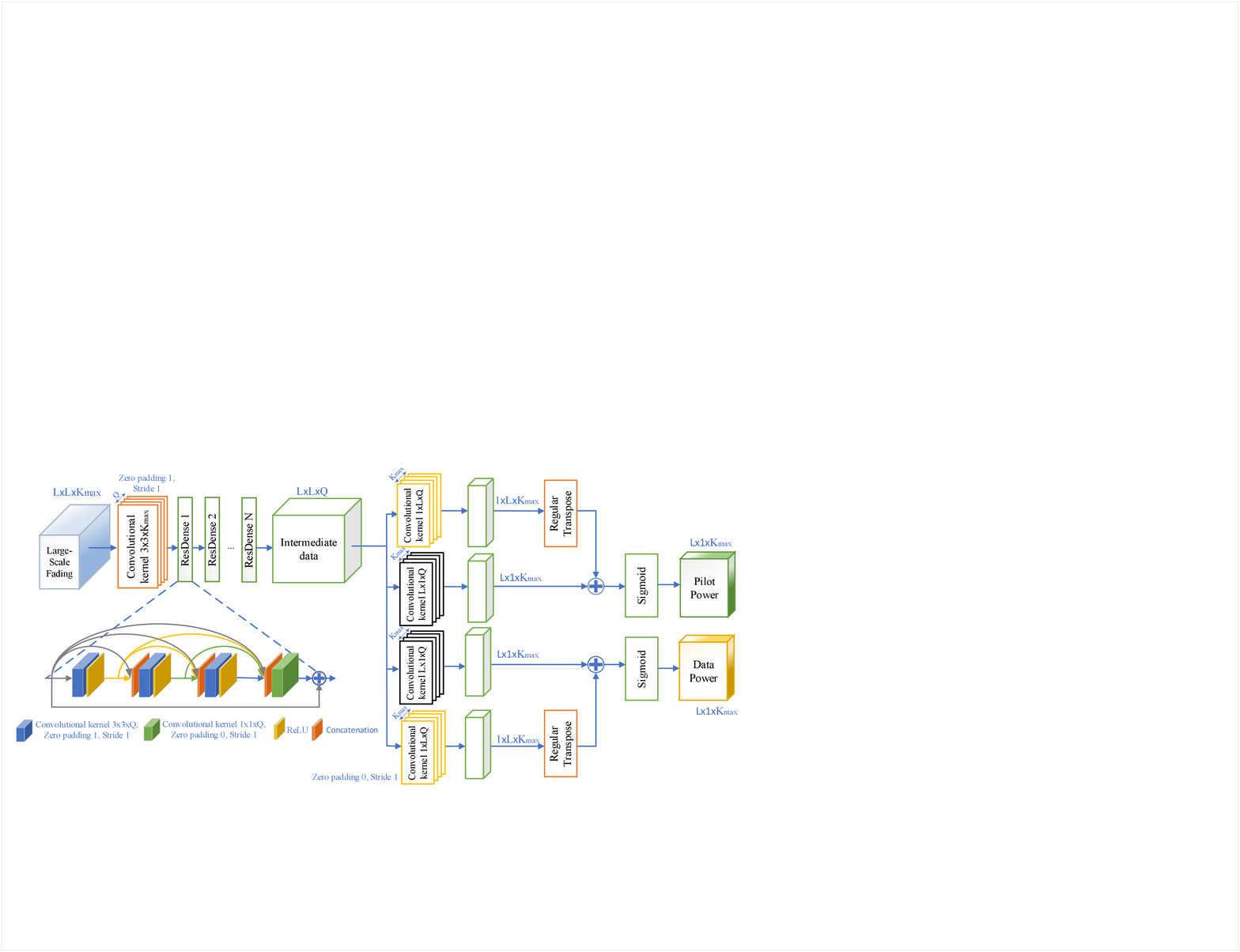} \vspace*{-0.1cm}
	\caption{The proposed PowerNet for the joint pilot and data power control from a given set of large-scale fading coefficients. }
	\label{FigCNN}
	\vspace*{-0.2cm}
\end{figure*}
\section{A low-complexity solution with convolutional neural network} \label{Section:CNN}

 In this section, we introduce a deep learning framework for joint pilot and data power allocation in dynamic cellular Massive MIMO systems, which uses supervised learning to mimic the power control obtained by Algorithm~\ref{Algorithm:AlternatingApproach}. We stress that for non-convex optimization problems, a supervised learning approach with high prediction accuracy is both useful for achieving a low-complexity implementation, harnessing the advances in implementing neural networks on GPUs, and provides a good baseline for further activities, e.g., supervised learning as a warm start for unsupervised learning or to improve the performance of the testing phase \cite{lee2018deep}.
 \subsection{Conditions on Large-Scale Fading Model}
 We first make an explicit assumption on how the large-scale fading coefficients for active users are generated for each realization of the Massive MIMO network.
 
 \begin{assumption} \label{assumption:UserLocation}
 	We consider an $L$-cell system where the activation of each user is determined by an independent Bernoulli distribution with activity probability $p\in [0,1]$. In each realization of the system, $K_{\max}$ i.i.d.~users are generated in each cell~$l$, in which user $t$ is active (i.e., $t \in \mathcal{A}_l$) with the probability $p$. The large-scale fading coefficients associated with an active user in cell $l$ have the probability density function (PDF) $f_l(\pmb{\beta})$, in which $\pmb{\beta} \in [0,1]^L$ and  $[\pmb{\beta}]_l = \max_{i \in \{1, \ldots, L\}} [\pmb{\beta}]_i$, for $l=1,\ldots,L$. 
 	The large-scale fading coefficients $\pmb{\beta}_{l,k} = [\beta_{l,k}^1, \ldots, \beta_{l,k}^L]^T$ of active user, $k \in \mathcal{A}_l$, is obtained as an i.i.d. realization with the PDF $f_l(\cdot)$ that satisfies 
 	\begin{equation} \label{eq:LSFCond}
 	\beta_{l,t}^l = \underset{i \in \{1, \ldots, L\}}{\max} \, \, \beta_{l,t}^i, \forall t \in \mathcal{A}_l,
 	\end{equation}
 	such that it has its strongest channel from the serving BS. Since the activity of the $K_{\max}$ users in cell~$l$ is randomly generated, there are $2^{K_{\max}}$ different possible realizations of the active index user set $\mathcal{A}_l$.
 \end{assumption}
 The process of generating system realizations is illustrated in Fig.~\ref{FigLFSDig}.
 Note that all users in cell $i$ have the same $f_i(\cdot)$, which represents the user distribution over the coverage area of this cell, but this function is different for each cell. For notational convenience, each cell has the same maximum number of users $K_{\max}$ and the activity probability is independent of the cell and location, but these assumptions can be easily generalized.
 
Assumption~\ref{assumption:UserLocation} indicates that a user should be handled equally irrespective of which number that it has in the cell. The fact that all large-scale fading coefficients belong the to compact set $[0,1]$ originates from the law of conservation of energy, and fits well with the structural conditions required to construct a neural networks \cite{hornik1989}. There are many ways to define the PDFs of the large-scale fading coefficients. One option is to match them to channel measurements obtained in a practical setup \cite{gao2013massive}. Another option is to define the BS locations and user distributions and then define a pathloss model with shadow fading. In the numerical part of this paper, we take the latter approach and follow the $3$GPP LTE standard \cite{LTE2010b} that utilizes a Rayleigh-lognormal fading model that matches well to channel measurements in non-line-of-sight conditions. The following model is used for simulations in Section~\ref{Section:NumericalResults}.
 
 \begin{example} \label{Example:LSF} 
 	Consider a setup with $L$ square cells. In each cell, the $K_{\max}$ users are uniformly distributed in the serving cell at distances to the serving BS that are larger than $35$\,m. Each user has the activity probability $p=2/3$. 
 	For an active user~$t$ in cell~$l$, we generate the large-scale fading coefficient to BS~$i$ as
 	\begin{equation} \label{eq:ShadowFading}
 	\beta_{l,t}^i\, \mathrm{[dB]} = -148.1 - 37.6\log_{10} (d_{l,t}^i / 1\,\mathrm{km})  + z_{l,t}^i, 
 	\end{equation}
 	where $d_{l,t}^i$ is the physical distance and $z_{l,t}^i$ is shadow fading that follows a normal distribution with zero mean and standard derivation $7$~dB. If the conditions \eqref{eq:LSFCond} and/or $\beta_{l,t}^i \leq 1$ are not satisfied for a particular user, we simply rerun all the shadow fading realizations for that user.
 \end{example}
 
In order to work with deep neural networks, we now use the channel model in Assumption~\ref{assumption:UserLocation} to express Algorithm~\ref{Algorithm:AlternatingApproach} with the viewpoint of continuous mappings.\footnote{The process $h(\mathbf{x}) = [h_1(\mathbf{x}),\ldots, h_{N_1}(\mathbf{x}) ]$ is a continuous mapping if all $h_n (\mathbf{x}), \forall n \in \{1, \ldots, N_1 \},$ are continuous functions.} Let $\mathcal{F}_{u}, \mathcal{F}_{w},  \mathcal{F}_{p},$ and $\mathcal{F}_{d}$ define the multivariate functions that handle the updates of variables $u_{l,k}, w_{l,k}, \hat{\rho}_{l,k},$ and $\rho_{l,k}, \forall l,k \in \mathcal{A}_l$, which are formally defined as 
 \begin{align}
 \{ u_{l,k}^{(n)} \}_{\forall l }^{k \in \mathcal{A}_l} &= \mathcal{F}_{u} \left(\{ \pmb{\psi}_{k} ,   \beta_{l,k}^{i} ,  \hat{\rho}_{l,k}^{(n-1)} , \rho_{l,k}^{(n-1)} \}_{\forall l,i }^{k \in \mathcal{A}_l} \right), \label{eq:Umap}\\
 \{w_{l,k}^{(n)}\}_{\forall l }^{k \in \mathcal{A}_l} &= \mathcal{F}_{w} \left( \{ u_{l,k}^{(n)} , \pmb{\psi}_k  ,  \beta_{l,k}^{i} , \hat{\rho}_{l,k}^{(n-1)} , \rho_{l,k}^{(n-1)} \}_{\forall l,i }^{k \in \mathcal{A}_l} \right), \\
 \{ \hat{\rho}_{l,k}^{(n)} \}_{\forall l }^{k \in \mathcal{A}_l} &= \mathcal{F}_{p} \left(\{ u_{l,k}^{(n)}  , w_{l,k}^{(n)},  \pmb{\psi}_k ,  \beta_{l,k}^{i} , \rho_{l,k}^{(n-1)} \}_{\forall l,i }^{k \in \mathcal{A}_l} \right),\\
 \{ \rho_{l,k}^{(n)} \}_{\forall l }^{k \in \mathcal{A}_l} &= \mathcal{F}_{d} \left(\{ u_{l,k}^{(n)}  ,w_{l,k}^{(n)} ,  \pmb{\psi}_k ,  \beta_{l,k}^{i}, \hat{\rho}_{l,k}^{(n)} \}_{\forall l,i }^{k \in \mathcal{A}_l} \right) \label{eq:Odmap},
 \end{align}
 The continuous mapping describing Algorithm~\ref{Algorithm:AlternatingApproach} at iteration~$n$ is now formulated in a compact form as 
 \begin{equation} \label{eq:ContMap}
 \begin{split}
& \{ \hat{\rho}_{l,k}^{(n)} , \rho_{l,k}^{(n)}  \}_{\forall l }^{k \in \mathcal{A}_l}  =  \mathcal{F}^{(n)} \left( \{ \pmb{\psi}_k , \beta_{l,k}^{i} , \hat{\rho}_{l,k}^{(n-1)}, \rho_{l,k}^{(n-1)} \}_{\forall l,i }^{k \in \mathcal{A}_l} \right) \\
&= \mathcal{F}^{(n)} \Big(\mathcal{F}^{(n-1)} \Big( \ldots \mathcal{F}^{(1)} \Big( \{ \hat{\rho}_{l,k}^{(0)} , \rho_{l,k}^{(0)} , \pmb{\psi}_k ,  \beta_{l,k}^{i} \}_{\forall l,i }^{k \in \mathcal{A}_l} \Big), \\
& \qquad \{ \pmb{\psi}_k , \beta_{l,k}^{i} \}_{\forall l,i }^{k \in \mathcal{A}_l} \Big), \{ \pmb{\psi}_k ,  \beta_{l,k}^{i} \}_{\forall l,i}^{k \in \mathcal{A}_l} \Big),
\end{split}
 \end{equation}
 where $\mathcal{F}^{(n)}$ is the continuous mapping representing for the whole process in \eqref{eq:Umap}--\eqref{eq:Odmap}, which is defined as $\mathcal{F}^{(n)} (\{ \pmb{\psi}_k ,  \beta_{l,k}^{i} , \hat{\rho}_{l,k}^{(n-1)}  , \rho_{l,k}^{(n-1)} \}_{\forall l,i }^{k \in \mathcal{A}_l} ) = \mathcal{F}_{d} \circ \mathcal{F}_{p} \circ \mathcal{F}_{w} \circ \mathcal{F}_{u} (\{ \pmb{\psi}_k , \beta_{l,k}^{i}, \hat{\rho}_{l,k}^{(n-1)} , \rho_{l,k}^{(n-1)} \}_{\forall l,i }^{k \in \mathcal{A}_l} ),$
 with $\circ$ being the composition operator. By assuming that Algorithm~\ref{Algorithm:AlternatingApproach} converges at iteration~$n$, i.e., $ \hat{\rho}_{l,k}^{\mathrm{opt}} =  \hat{\rho}_{l,k}^{(n)}, \rho_{l,k}^{\mathrm{opt}} =  \rho_{l,k}^{(n)}, \forall l, k \in \mathcal{A}_l,$ then the end-to-end continuous mapping for a given set of active users standing for the iterative processing of this algorithm is 
 \begin{equation} \label{eq:Mappingv1}
 \mathcal{F}\left(\{ \pmb{\psi}_k , \beta_{l,k}^{i} ,  \rho_{l,k}^{(0)} ,  \hat{\rho}_{l,k}^{(0)} \}_{\forall l,i }^{k \in \mathcal{A}_l} \right) = \left\{ \hat{\rho}_{l,k}^{\mathrm{opt}} ,  \rho_{l,k}^{\mathrm{opt}} \right\}_{\forall l,i }^{k \in \mathcal{A}_l}, 
 \end{equation}
 where 
 \begin{equation}
 \begin{split}
 &\mathcal{F} \left(\{ \pmb{\psi}_k , \beta_{l,k}^{i}, \rho_{l,k}^{(0)} , \hat{\rho}_{l,k}^{(0)} \}_{\forall l,i }^{k \in \mathcal{A}_l} \right) =\\
 & \mathcal{F}^{(n)} \left(\{ \pmb{\psi}_k ,  \beta_{l,k}^{i} , \rho_{l,k}^{(n-1)}  ,  \hat{\rho}_{l,k}^{(n-1)} \}_{\forall l,i }^{k \in \mathcal{A}_l} \right)
 \end{split}
\end{equation}
that shows a stationary point obtained from the input set of large-scale fading together with an initial set of powers and the pilot assignment after $n$ iterations. We stress that the input of the continuous mapping \eqref{eq:Mappingv1} varies as a function of the user activity per cell, i.e. $\mathcal{A}_l, \forall l$, thus there are many different power control functions that would need to be approximated with individual neural networks if one follows the standard approach taken in previous work \cite{Luca2018DL, sun2018learning}. A single neural network cannot have a varying number of inputs and outputs, so we would have to design a separate neural network for each value of $\mathcal{A}_l$. If the methods in previous works \cite{Luca2018DL, sun2018learning} are applied for the considered scenario in this paper, BS~$l$ must construct and train $2^{K_{\max}}$ separate neural networks to learn all subset relations created by $\mathcal{A}_l$ as a consequence of Assumption~\ref{assumption:UserLocation}. For a cellular network comprising $L$ cells, we will need to train and store $L2^{K_{\max}}$ different neural networks, which is a big number (up to $9126$ neural networks in the simulation part with $L= 9$ cells serving $K_{\max} = 10$ users). If we had to design all such specific neural networks, the solution is practically useless.

\subsection{Existence of a Single Neural Network for Joint Pilot and Data Power Control}
A main contribution of our framework is that we can exploit the structure of the optimization problem to build a single neural network that can handle the activity/inactivity pattern and has a unified structure for all training samples. We exploit Corollary~\ref{corollary:phatlkandplk} to design such a single neural network that approximates the whole family of functions in \eqref{eq:Mappingv1}. Recall that this corollary says that users with zero-valued large-scale fading coefficients are assigned zero power, without affecting the other users in the network. Corollary~\ref{corollary:phatlkandplk} is indeed used to remove the number of active users per cell and the pilot assignment as input variables. These variables are instead implicitly represented by the large-scale fading values, which are now vectors of a predefined size where non-existing users are represented by zeros, i.e., $\pmb{\beta}_{l,k} = [\beta_{l,k}^1, \ldots, \beta_{l,k}^L]^T = \mathbf{0}, \forall k \notin \mathcal{A}_l$. After that, we define a tensor $\mathsf{I} \in \mathbb{R}_{+}^{L \times L \times K_{\max}}$ containing all the large-scale fading coefficients, $\mathsf{O}_d^{\mathrm{opt}} \in \mathbb{R}_{+}^{L \times 1 \times K_{\max}}$ containing the optimized data powers, and $\mathsf{O}_p^{\mathrm{opt}} \in \mathbb{R}_{+}^{L \times 1 \times K_{\max}}$ containing the optimized pilot powers. Following similar steps as presented in \eqref{eq:Umap}--\eqref{eq:Mappingv1} and treating the starting point as the predetermined number, PowerNet will learn the unique continuous mapping
\begin{equation} \label{eq:Mapping}
\mathcal{F}(\mathsf{I}) = \left\{ \mathsf{O}_d^{\mathrm{opt}} , \mathsf{O}_p^{\mathrm{opt}} \right\}.
\end{equation}
The key distinction of the continuous mapping \eqref{eq:Mapping} compared with one in \eqref{eq:Mappingv1} is the fixed number of inputs and outputs, irrespective of the user activity. This allows a single deep neural network to handle the variation in user activity. We can view the solution to a power control optimization problem as a power control mapping. It is well known that deep learning can be used to approximate such a mapping, thanks to the universal approximation theorem \cite{hornik1989}. The input to the proposed feed-forward neural network is only the large-scale fading coefficients and the output is the data and pilot powers. This makes PowerNet more applicable for real-time communication systems than the previous works \cite{zappone2018model, sun2018learning}, which rely on perfect instantaneous CSI (i.e., small-scale fading) to predict the data power allocation, and neglected channel estimation and pilot contamination effects.
\begin{remark} \label{Remark:ProblemStructure}
The zero-insertion trick that was used to achieve~\eqref{eq:Mapping} works well thanks to the structure of our sum SE optimization problem as stated in Corollary~\ref{corollary:phatlkandplk}. The same trick cannot be applied along with previous works, for example \cite{Luca2018DL}, when using the maximum product-SINR or max-min fairness power control as the performance metric since then everyone would receive zero powers.
\end{remark}

\vspace*{-0.25cm}
\subsection{Convolutional Neural Network Architecture with Residual Dense Connections} 
Among all neural network structures in the literature, CNN is a popular family that achieves higher performance than fully-connected deep neural network for many applications \cite{sainath2013, zhang2017beyond}. One main reason reported in \cite{sainath2013} is that CNN effectively deduces the spectral variation existing in a dataset. In order to demonstrate why the use of CNN is suitable for power control in Massive MIMO, let us consider an area of $5 \times 5$~km with $L = 256$ square cells, each serving $K_{\max} = 10$ users. The large-scale fading coefficients are generated as in Example~\ref{Example:LSF}, but all users are assumed to be in active mode. The interference in a real cellular system is imitated by wrap-around. We gather all the large-scale fading coefficients in a tensor of size $L \times L \times K_{\max}$. The main reason for using a CNN is that the third dimension contains  $K_{\max}$ users that have i.i.d.~coefficients, which are processed identically when using convolutional kernels. An additional reason is that the geometry of the BS deployment can be learned and utilized. For visualization, we map this tensor to a matrix $\mathbf{Z}$ of size $L \times L$ by averaging over the third dimension (user dimension) and plot the result in Fig.~\ref{FigPatter}. The number of horizontal and vertical elements is equal to $L$.

The color map in Fig.~\ref{FigPatter} represents the large-scale fading coefficients. For example, the color of the square $(l,j)$ represents the average large-scale fading coefficient from a user in cell~$l$ to BS~$j$.\footnote{This figure considers a symmetric network, where the BSs are deployed on a square grid, to make the pattern clearly visible to the reader. We consider a network with the same shape in the simulation part but with fewer BSs. We have observed that a CNN can identify patterns in symmetric deployments since the BS locations follow a strict symmetric pattern. The design of a neural network that is capable of identifying similar patterns in networks with asymmetric BS locations is left for the future.}
 Since there is a grid of $16 \times 16$ cells, and the cells are numbered row by row, the large-scale fading coefficients have a certain pattern determined by the fixed cell geometry. Users in neighboring cells have larger large-scale fading coefficients than cells that are further away. The strong intensity around the main diagonal represent the cell itself, while the sub-diagonals with strong intensities represent interference between neighboring cells. The strong intensities in the lower-left and upper-right corners are due to the wrap-around topology. In summary, we make the following observations:
\begin{itemize}
\item The elements of the $L \times L \times K_{\max}$ tensor with large-scale coefficients have a particular relation determined by the geometry of the cellular network. The neural network must learn these relationship and the neural network structure should be selected to facilitate this learning. A CNN is appropriate since the tensor can be directly used as input.\footnote{According to the universal approximation theorem, it is possible to design a fully-connected neural network (FDNN) to achieve the same prediction accuracy as any other neural network. However, we stress that the tensor with large-scale fading coefficients must be turned into a vector in that case. The network will be overparameterized since many weights and biases will do the same thing (e.g., treated users in the same cell) and is becomes harder for the neural network to identify the structure that we know exist in the input data. This intuition is confirmed by Figs.~\ref{FigSEperCellL4K10Ber2p3} and \ref{FigPerUserSEL4K10Ber2p3}, where we compare a CNN with an FDNN.}

\item Each of the $K_{\max}$ users in a cell have large-scale fading coefficients generated independently from the same distribution, thus these inputs should be treated equally.
It is therefore beneficial to have a set convolutional kernels that are reused among the users in a cell, instead of an overparameterized fully-connected network where we know that some of the weights and biases should be the same when we have trained the network properly.

\item Since the cell geometry and BS deployment is fixed in practical deployment, the CNN can learn recurrent spatial patterns of the type shown in Fig.~\ref{FigPatter}. In practice, the patterns might not be as visible to the human eye due to asymmetric cell deployment, but machine learning can anyway identify and exploit it to enable reuse of a set convolutional kernels that are reused in multiple cells, thereby reducing the number of trainable parameters.

\end{itemize}
This paper designs a CNN (illustrated in Fig.~\ref{FigCNN}) which is able to obtain performance close to the stationary point produced by the iterative algorithm in Theorem~\ref{Theorem:IterativeAl}. This CNN only uses the large-scale fading coefficients gathered in an $L \times L \times K_{\max}$  tensor as input, which reduces significantly the signaling overhead compared to using instantaneous channels  as in \cite{van2018distributed}, and reduces the number of trainable parameters by exploiting the pattern created by the large-scale fading coefficients. We will further adopt the state-of-the-art residual dense block (ResDense) \cite{zhang2018a} which consists of densely connected convolutions \cite{huang2017} with the residual learning \cite{he2016deep}. The residual dense connections, which are demonstrated by concatenations, prevent the vanishing gradient problem that may often happen when some large-scale fading coefficients are small. It also prevents the overfitting problem by using many ResDense blocks to deploy a very deep neural network. Compared with the original ResDense in \cite{zhang2018a}, we use additional (rectified linear unit) ReLU activation unit, $\xi(x) = \max(0,x)$, after the residual connection since our mapping only concentrates on non-negative values. We notice that the proposed network might have more parameters than actually needed since our main goal is to provide a proof-of-concept. The optimal network with the lowest number of parameters is different for every propagation environment and therefore not considered in this work, which focuses on the general properties and not the fine-tuning.

\subsubsection{The forward propagation} From the large-scale fading tensor $\mathsf{I}$, the first component of the forward propagation is the convolutional layer 
\begin{equation} \label{eq:1stFeatureMap}
\mathsf{X}_1^{(m)} = \mathcal{H}_1 \left(\mathsf{I}, \{ \mathsf{W}_{1, j}^{(m-1)}, b_{1, j}^{(m-1)} \}_{j=1}^Q \right),
\end{equation}
where $m$ is the epoch index. The operator $\mathcal{H}_1(\cdot, \cdot)$ denotes a series of $Q$ convolutions \cite{chen2016deep}, each using a kernel $\mathsf{W}_{1, j}^{(m-1)}  \in \mathbb{R}^{3\times 3\times K_{\max}}$ and a bias $b_{1, j}^{(m-1)} \in \mathbb{R}$ to extract large-scale fading features of the input tensor $\mathsf{I}$. 
All convolutions apply stride~$1$ and zero padding $1$ to guarantee the same height and width between the inputs and outputs.  After the first layer in \eqref{eq:1stFeatureMap}, the feature map is a tensor with the size $L\times L \times Q$. PowerNet is constructed from $N$ sequential connected ResDense blocks to extract special features of large-scale fading coefficients.  Each ResDense block uses the four sets of convolutional kernels to extract better propagation features. The first convolution begins with $\mathsf{X}_{2,1,4}^{(m)} = \mathsf{X}_1^{(m)}$, then the output signal at each of the $n$-th ResDense block is simultaneously computed by adapting the residual dense connection in \cite{he2016deep} to PowerNet as:
\begin{align}
\mathsf{X}_{2,n,i}^{(m)} &= \xi \left(\mathcal{H}_{2, i} \left( [\mathsf{X}_{2,n-1,4}^{(m)}, \mathsf{X}_{2,n,1}^{(m)}, \ldots, \mathsf{X}_{2,n,i-1}^{(m)} ], \{ \mathsf{W}_{2, i, j}^{(m-1)} \}_{j=1}^Q \right) \right), \label{eq:Resden1} \\
\mathsf{X}_{2,n,4}^{(m)} &= \mathcal{H}_{2, 4} \left([\mathsf{X}_{2,n-1,4}^{(m)}, \mathsf{X}_{2,n,1}^{(m)}, \mathsf{X}_{2,n,2}^{(m)}, \mathsf{X}_{2,n,3}^{(m)} ], \{ \mathsf{W}_{2, 4, j}^{(m-1)} \}_{j=1}^Q \right),\label{eq:Resden4}
\end{align}
where $i \in \{1,2,3 \}$ in \eqref{eq:Resden1}. Each operator $\mathcal{H}_{2,j} (\cdot, \cdot), j \in \{1,\ldots,4\},$ denotes a series of the $Q$ convolutions. In the three first modules, the kernels are $\mathsf{W}_{2, 1, j}^{(m-1)} \in \mathbb{R}^{3\times 3 \times Q}$, $\mathsf{W}_{2, 2, j}^{(m-1)} \in \mathbb{R}^{2 \times 3\times 3 \times Q}$, $\mathsf{W}_{2, 3, j}^{(m-1)} \in \mathbb{R}^{3 \times 3\times 3 \times Q}$, while the remaining has $\mathsf{W}_{2, 4, j}^{(m-1)} \in \mathbb{R}^{4 \times 1\times 1 \times Q}$. In the first three modules, the ReLU activation function $\xi(x)$ is used for each element. 

Since the input and output size of the neural network are different, multiple $1$D convolutions are used to make the sides equal. In addition, both horizontal and vertical 1D convolutions are used to exploit correlation in both directions. A regular transpose layer is applied following vertical 1D convolution to ensure the data size of $L \times 1 \times K_{\max}$. This prediction is used for both pilot and data power as depicted in Fig.~\ref{FigCNN} and is mathematically expressed as
\begin{align} 
\mathsf{X}_{z}^{(m)} = & \mathcal{H}_{ z}^v \left( \widetilde{\mathsf{X}}_{2,N,4}^{(m)} , \{ \mathsf{W}_{ z,j}^{v,(m-1)} , b_{ z,j}^{v,(m-1)} \}_{j=1}^{K_{\max}} \right) + \notag \\
& 
\mathcal{H}_{ z}^h \left( \widetilde{\mathsf{X}}_{2,N,4}^{(m)} , \{ \mathsf{W}_{ z,j}^{h,(m-1)}, b_{ z,j}^{h,(m-1)} \}_{j=1}^{K_{\max}} \right), \label{eq:PilotTensor}
\end{align}
where the index $z$ stands for either the pilot or data tensor; $\widetilde{\mathsf{X}}_{2,N,4}^{(m)} = \mathsf{X}_{2,N,4}^{(m)} + \mathsf{X}_{1}^{(m)}$; $\mathcal{H}_{z}^v (\cdot, \cdot)$ and $\mathcal{H}_{ z}^h (\cdot, \cdot)$ denote the vertical and horizontal series of $K_{\max}$ convolution operators used to predict the optimized powers by convolutional kernels $\mathsf{W}_{ z,j}^{v,(m-1)} \in \mathbb{R}^{1 \times L \times Q}, \mathsf{W}_{ z,j}^{h,(m-1)} \in \mathbb{R}^{L \times 1 \times Q}, \forall j$ and their related biases $b_{z,j}^{v,(m-1)}, b_{ z,j}^{h,(m-1)} \in \mathbb{R}, \forall j$. 
The feature maps from \eqref{eq:PilotTensor} are restricted in the closed unit interval $[0,1]$ by 
\begin{equation}
\mathsf{X}_{z,s}^{(m)} = \mathrm{Sigmoid}\left( \mathsf{X}_{z}^{(m)}  \right), 
\end{equation}
where the element-wise sigmoid activation function is $\mathrm{Sigmoid} (x) = 1/(1 + \exp(-x))$. Finally, the predicted pilot and data powers at epoch~$m$ are obtained by scaling up $\mathsf{X}_{z,s}^{(m)}$ and $\mathsf{X}_{z,s}^{(m)}$ as
\begin{equation}
\mathsf{O}_z^{(m)} =  \mathsf{P}\odot  \mathsf{X}_{z,s}^{(m)}, 
\end{equation} 
where $\mathsf{P} \in \mathbb{R}^{L\times 1 \times K_{\max}}$ is a collection of the maximum power budget $P_{l,k}$ from all users with $[\mathsf{P}]_{l,1,k} = P_{l,k}, \forall l,k$ and $\odot$ denotes the dot product of two tensors. The forward propagation is applied for both the training and testing phases.

\subsubsection{The back propagation} \label{Sec:BackPro}
The back propagation is only applied in the training phase. We consider supervised learning where the Frobenius norm is adopted to define the loss function as
\begin{equation} \label{eq:Loss}
f (\Theta^{(m)} ) =   \frac{1}{D} \sum_{i=1}^D \big\|  \mathsf{O}_{d}^{(m),i} - \mathsf{O}_{d}^{\mathrm{opt},i}  \big\|_F^2  +  \frac{1}{D} \sum_{i=1}^D \big\| \mathsf{O}_{p}^{(m),i}  - \mathsf{O}_{p}^{\mathrm{opt},i} \big\|_F^2
\end{equation}
with respect to the convolution kernels and biases in $\Theta$. The loss in \eqref{eq:Loss} is averaged over the training dataset $\{ \mathsf{O}_{p}^{\mathrm{opt}, i}, \mathsf{O}_{d}^{\mathrm{opt}, i} \}, i = 1, \ldots, D,$ where $D$ is the total number of large-scale fading realizations. While the supervised learning provides a simple loss function for the back propagation, it make PowerNet on average perform not better than the continuous mapping~\eqref{eq:Mapping}.

The back propagation utilizes \eqref{eq:Loss} to update all weights and biases. PowerNet will use stochastic gradient descent \cite{Goodfellow-et-al-2016} to obtain a good local solution to $\Theta$. Beginning with a random initial value $\Theta = \Theta^{(0)} , \Delta \Theta^{(0)} = \mathbf{0},$ and remember the current $ \Delta  \Theta^{(m)}$, epoch $m$ updates $\Theta^{(m)}$ 
using the so-called momentum $\alpha$ and the learning rate $\eta$. The computational complexity of the back propagation can be significantly reduced, while still can obtain a good solution to the kernel and bias if a random mini-batch $D_t$ with $D_t < D$ is properly selected \cite{Goodfellow-et-al-2016} rather than processing all the training data at once.

\begin{remark} \label{Remark:BestLabel}
	Since Algorithm~\ref{Algorithm:AlternatingApproach} only finds a stationary point, different initializations will lead to different stationary points, some giving a higher sum SE than others. Instead of creating the labels by running Algorithm~\ref{Algorithm:AlternatingApproach} only once, we can run the algorithm multiple times with different initializations and use the one with the highest sum SE as the label. The more initializations are used, the more likely it is that the global optimum is found. This method only increases the computational complexity to generate good data labels, while PowerNet retains the computational complexity as before. We will evaluate the performance benefit in Section~\ref{Section:NumericalResults}.
\end{remark} 

\vspace*{-0.5cm}
\subsection{Dataset, Training, and Testing Phases}

In order to train PowerNet, we use Algorithm~\ref{Algorithm:AlternatingApproach} to generate training pairs of user realizations and the corresponding outputs $\mathsf{O}_{p}^{\mathrm{opt}}, \mathsf{O}_{d}^{\mathrm{opt}}$ that are jointly optimized by our method presented in Algorithm~\ref{Algorithm:AlternatingApproach}. Aligned with the previous works \cite{Luca2018DL, lee2018deep}, an offline training phase is considered in this paper.\footnote{Future work can explore how to train the network using less training samples, for example, using semi-supervised methods.} Specifically, we generate data with the mini-batch size $L \times L \times K_{\max} \times D_t$ for the training and testing phase, respectively. We use the momentum and babysitting the learning rate to get the best prediction performance and minimize the training time as well. The Adam optimization is used to train our data set \cite{kingma2014adam}. PowerNet is dominated by exponentiations, divisions, and multiplications, the number of arithmetic operations required for the forward propagation at each epoch is computed as
\begin{equation}
9K_{\max}L^2Q  + 28 Q^2 L^2 N + 4 L^2 Q K_{\max} + 4LK_{\max},
\end{equation}
which is also the exact computational complexity of the testing phase where each large-scale fading tensor only passes through the neural network once and there is no back propagation. 

\vspace*{-0.25cm}

\begin{figure*}[t]
	\begin{minipage}{0.49\textwidth}
		\centering
		\includegraphics[trim=0.6cm 0.0cm 1.1cm 0.7cm, clip=true, width=3.2in]{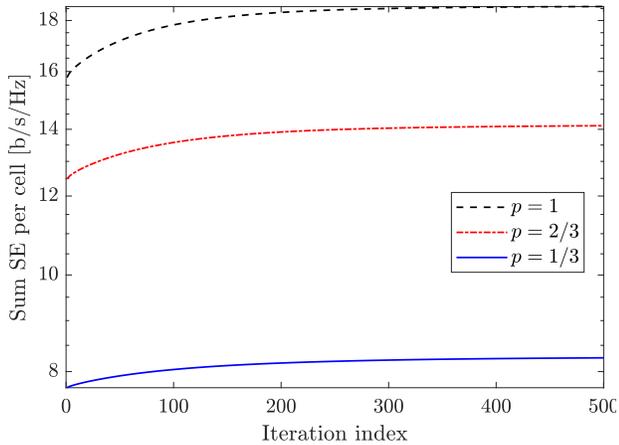} \vspace*{-0.1cm} \\
		(a)
	\end{minipage}
	\hfill
	\begin{minipage}{0.49\textwidth}
		\centering
		\includegraphics[trim=0.6cm 0.0cm 1.1cm 0.7cm, clip=true, width=3.2in]{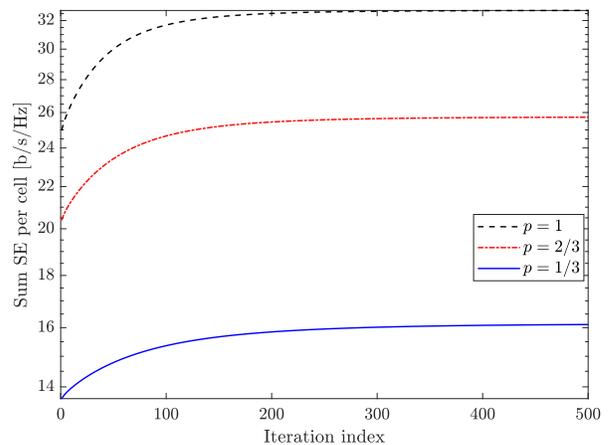} \vspace*{-0.1cm} \\
		(b)
	\end{minipage}
	\vspace*{-0.1cm}
	\caption{The convergence of Algorithm~\ref{Algorithm:AlternatingApproach} for the different activity probabilities: Fig.~\ref{FigConvg}a: the system with $L=4, K_{\max}= 5$. Fig.~\ref{FigConvg}b: The system with $L=9, K_{\max} =10$. }
	\label{FigConvg}
	\vspace*{-0.2cm}
\end{figure*}
\section{Numerical Results} \label{Section:NumericalResults}
To demonstrate the performance of PowerNet, we consider the setup in Example $1$ with $L \in \{4, 9\}$ equally large square cells in a square area $1$~km$^2$. The coverage area is wrapped around at the edges to avoid boundary effects, so $9L$ cells are considered in total. Each user is served by the BS having the largest large-scale fading coefficient.  In every cell, the BS is located at the center. The distribution of users and the large-scale fading coefficients are generated according to Example~\ref{Example:LSF}, but the activity probability will be defined later. The maximum power level is $P_{l,k} =200$~mW, $\forall l,k$. We use Algorithm~\ref{Algorithm:AlternatingApproach} to generate $1.5$ million data samples to train PowerNet. The mini-batch size is $512$. The number of epochs used for the training phase is $300$. We use a momentum of $0.99$ and babysitting of the learning rate which varies from $10^{-3}$ to $10^{-5}$. From our experiments, we note that the learning rate may be reduced by approximately three times if the test loss remains the same for $100$ consecutive epochs. \textcolor{red}{In the first convolutional layer, $64$ kernels are used and PowerNet has $5$ ResDense blocks}. For the loss function in \eqref{eq:Loss}, we set $w_1 = w_2 =1$ to treat the importance of the data and pilot powers equally. The following methods are compared: 
\begin{itemize}[leftmargin=4mm]
	\item[$1)$] \textit{Fixed power (FP) level}: Each user uses the fixed maximum power level $200$ mW for both pilot and data. It is denoted as FP in the figures. 
	\item[$2)$] \textit{Data power optimization only (DPOO)}: The system uses a simplification of Algorithm \ref{Algorithm:AlternatingApproach} to perform data power control, while the pilot power is fixed to $200$ mW. It is denoted as DPOO in the figures. 
	\item[$3)$] \textit{Joint pilot and data power optimization (JPDPO)}: The system uses Algorithm \ref{Algorithm:AlternatingApproach} to jointly optimize the optimal pilot and data powers for all users. It is denoted as JPDPO in the figures.
	\item[$4)$] \textit{Joint pilot and data power optimization based on CNN (PowerNet)}: The system uses the proposed CNN described in Section~\ref{Section:CNN} to find the pilot and data powers for all users. It is denoted as PowerNet in the figures. 
	\item[$5)$] \textit{Joint pilot and data power optimization based on fully-connected deep neural network}: The system uses a modified version of the fully-connected deep neural network in \cite{sun2018learning} to find the solution to both the pilot and data powers for all users. It is denoted as FDNN in the figures. 
\end{itemize}
\begin{figure*}[t]
	\begin{minipage}{0.49\textwidth}
		\centering
		\includegraphics[trim=0.5cm 0.0cm 0.5cm 0.7cm, clip=true, width=3.2in]{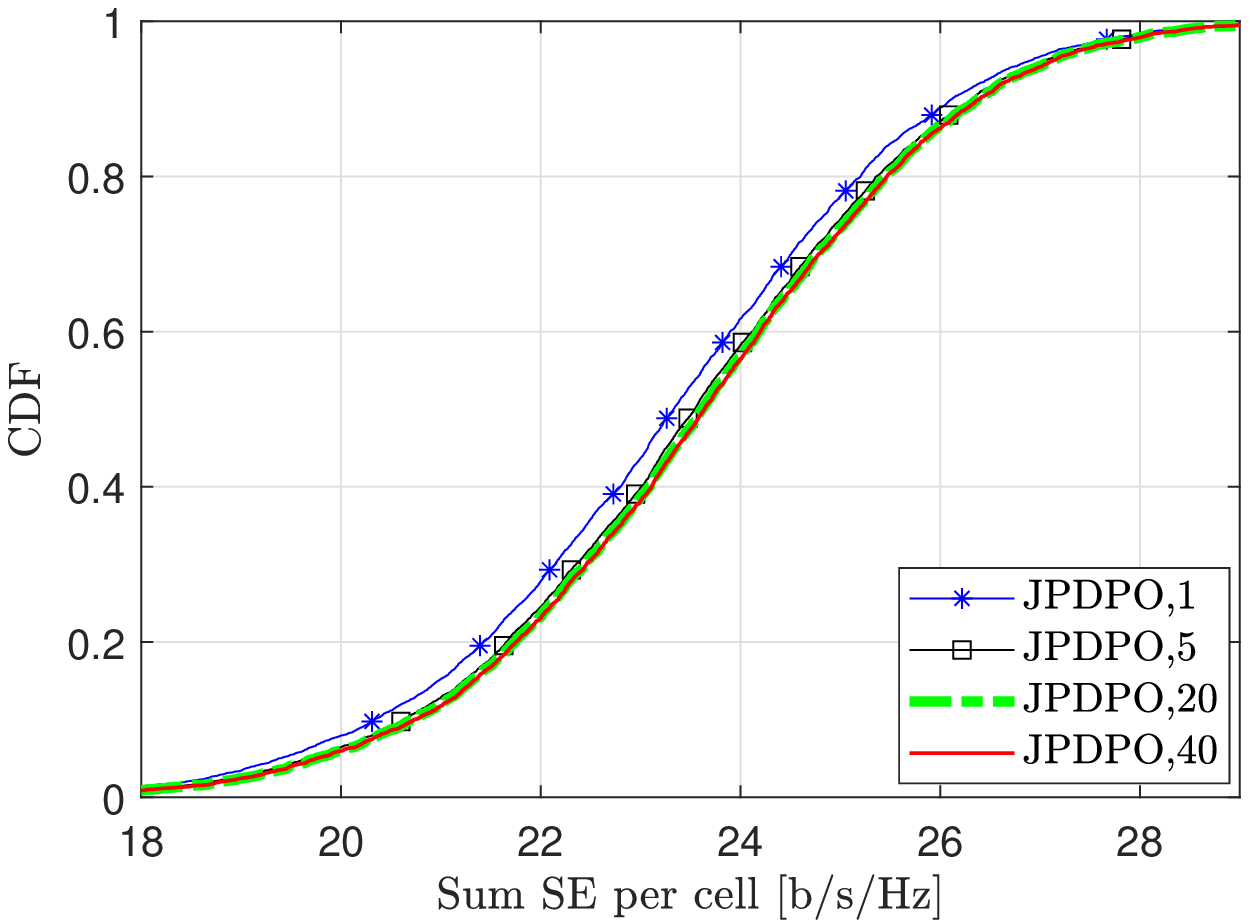} \vspace*{-0.1cm}
		\caption{CDF of sum SE per cell [b/s/Hz] versus the number of random initializations that we select the best out of. We have $L=4$, $K_{\max} = 10, M=200,$ and activity probability $2/3$.}
		\label{FigCDFDiffInitialization}
		\vspace*{-0.5cm}
	\end{minipage}
	\hfill
	\begin{minipage}{0.49\textwidth}
		\centering
		\includegraphics[trim=0.5cm 0.0cm 0.5cm 0.7cm, clip=true, width=3.2in]{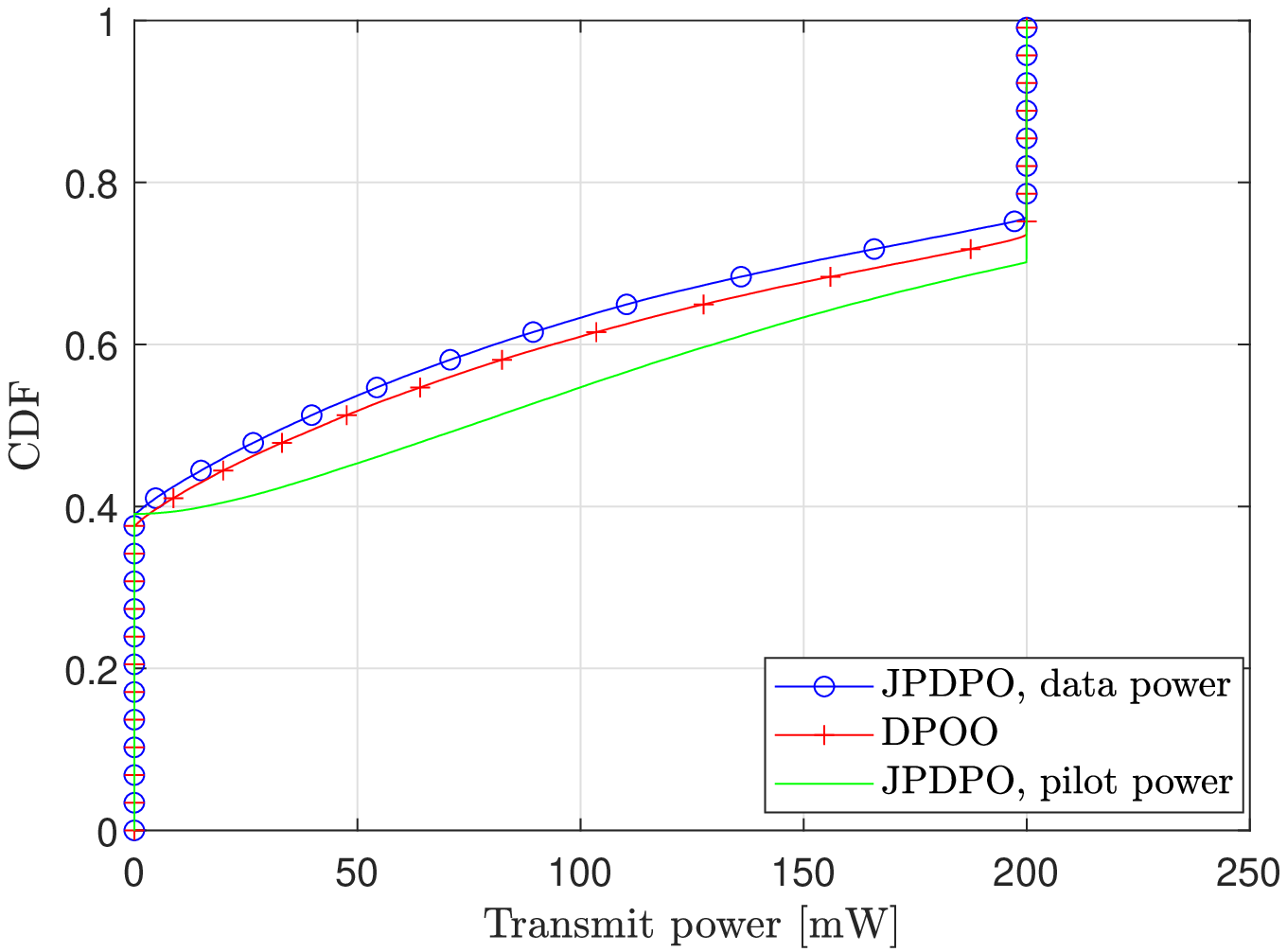} \vspace*{-0.1cm}
		\caption{CDF of pilot and data power allocation [mW] by using JPDPO and DPOO with $L=4, K_{\max} = 10, M=200,$ and activity probability $2/3$.}
		\label{FigPowerL4K10}
		\vspace*{-0.2cm}
	\end{minipage}
	
\end{figure*}

\begin{figure*}[t]
	\begin{minipage}{0.49\textwidth}
		\centering
		\includegraphics[trim=0.5cm 0.0cm 0.5cm 0.7cm, clip=true, width=3.2in]{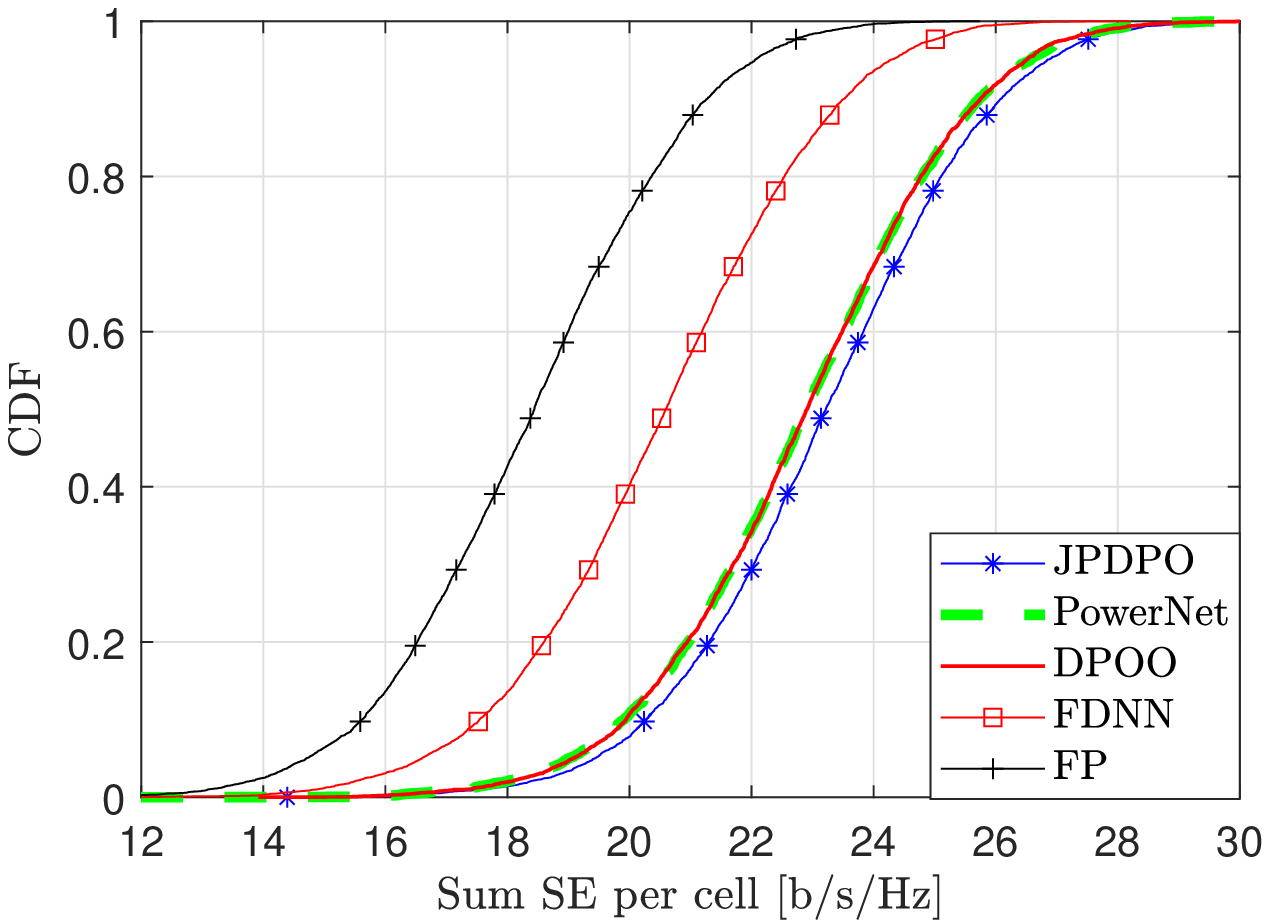} \vspace*{-0.1cm}
		\caption{CDF of sum SE per cell [b/s/Hz] with $L=4, K_{\max} = 10,M=200,$ and activity probability $2/3$.}
		\label{FigSEperCellL4K10Ber2p3}
		\vspace*{-0.2cm}
	\end{minipage}
	\hfill
	\noindent\begin{minipage}{0.49\textwidth}
		\centering
		\includegraphics[trim=0.5cm 0.0cm 0.5cm 0.7cm, clip=true, width=3.2in]{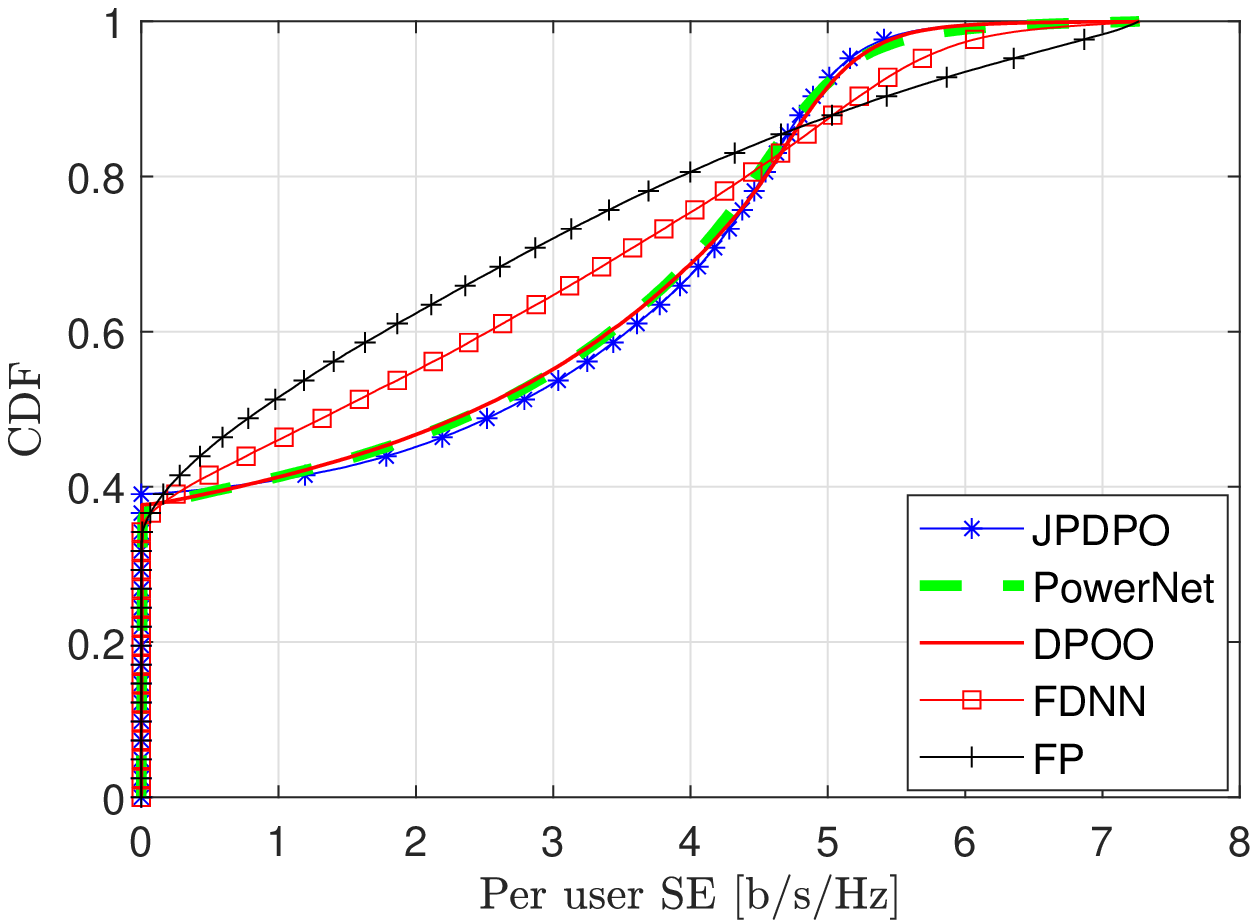} \vspace*{-0.1cm}
		\caption{CDF of per user SE [b/s/Hz] with $L=4, K_{\max} = 10,M=200,$ and activity probability $2/3$.}
		\label{FigPerUserSEL4K10Ber2p3}
		\vspace*{-0.2cm}
		\noindent \end{minipage}
\end{figure*}

\begin{figure*}[t]
	
	\begin{minipage}{0.49\textwidth}
		\centering
		\includegraphics[trim=0.5cm 0.0cm 0.5cm 0.7cm, clip=true, width=3.2in]{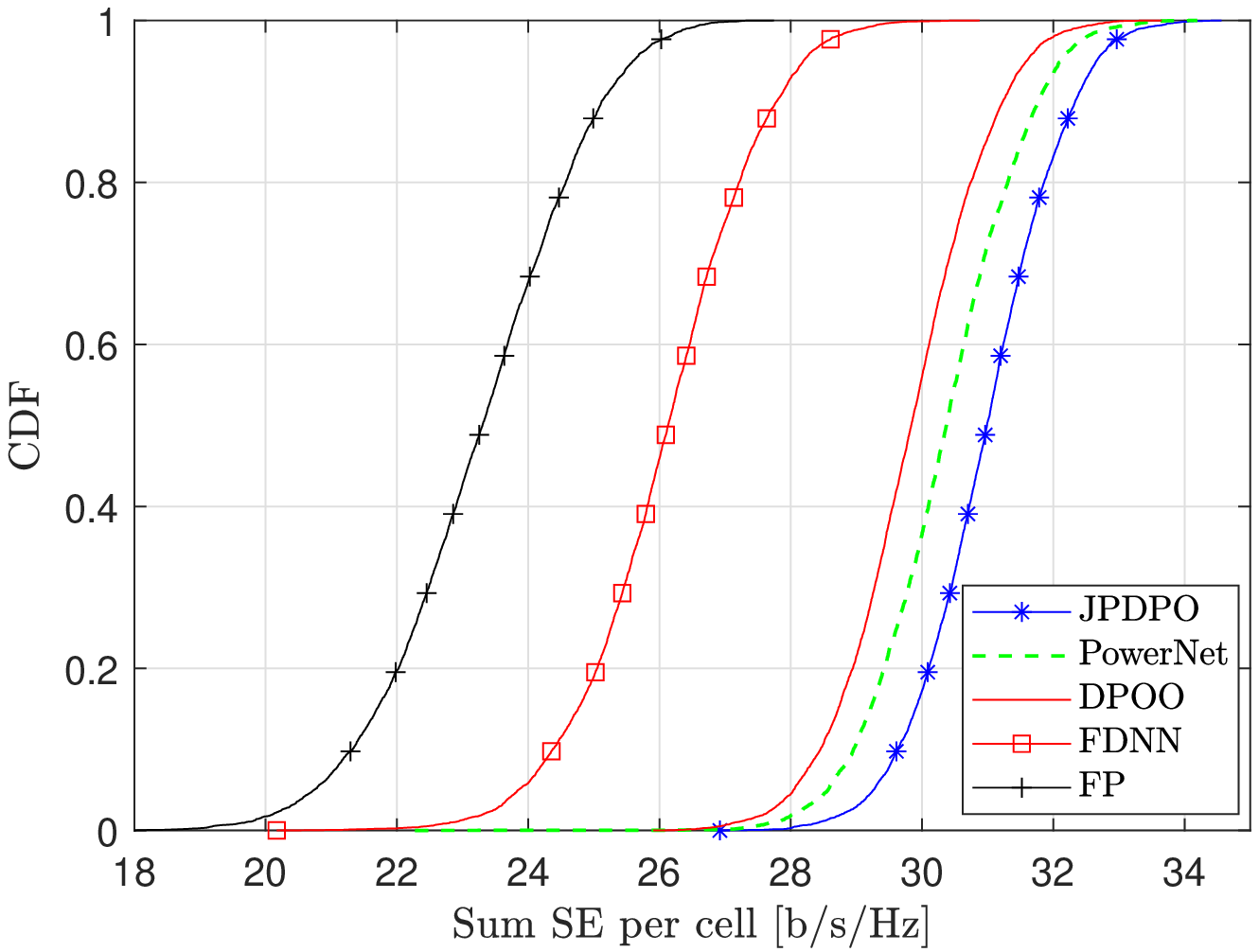} \vspace*{-0.1cm}
		\caption{CDF of sum SE per cell [b/s/Hz] with $L=9, K_{\max} = 10,M=200,$ and activity probability $1$.}
		\label{FigSumSEperCellL9K10}
		\vspace*{-0.2cm}
	\end{minipage}
	\hfill
	\begin{minipage}{0.49\textwidth}
		\centering
		\includegraphics[trim=0.5cm 0.0cm 0.5cm 0.7cm, clip=true, width=3.2in]{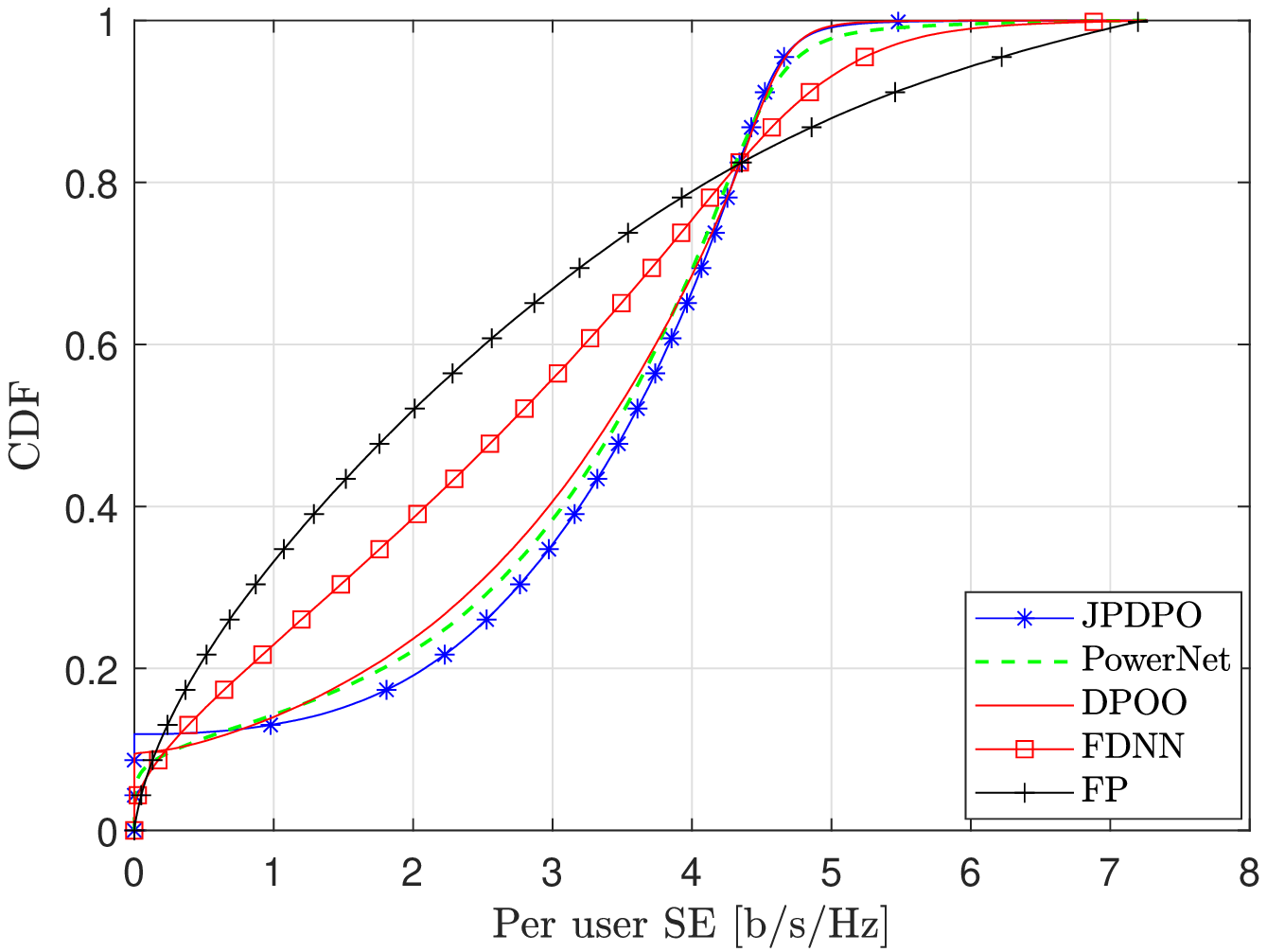} \vspace*{-0.1cm}
		\caption{CDF of per user SE [b/s/Hz] with $L=9, K_{\max} = 10,M=200,$ and activity probability $1$.}
		\label{FigCDFPerUserRateL9K10}
		\vspace*{-0.2cm}
	\end{minipage}
\end{figure*}

\begin{figure*}[t]
	\noindent\begin{minipage}{0.49\textwidth}
		\centering
		\includegraphics[trim=0.5cm 0.0cm 0.5cm 0.7cm, clip=true, width=3.2in]{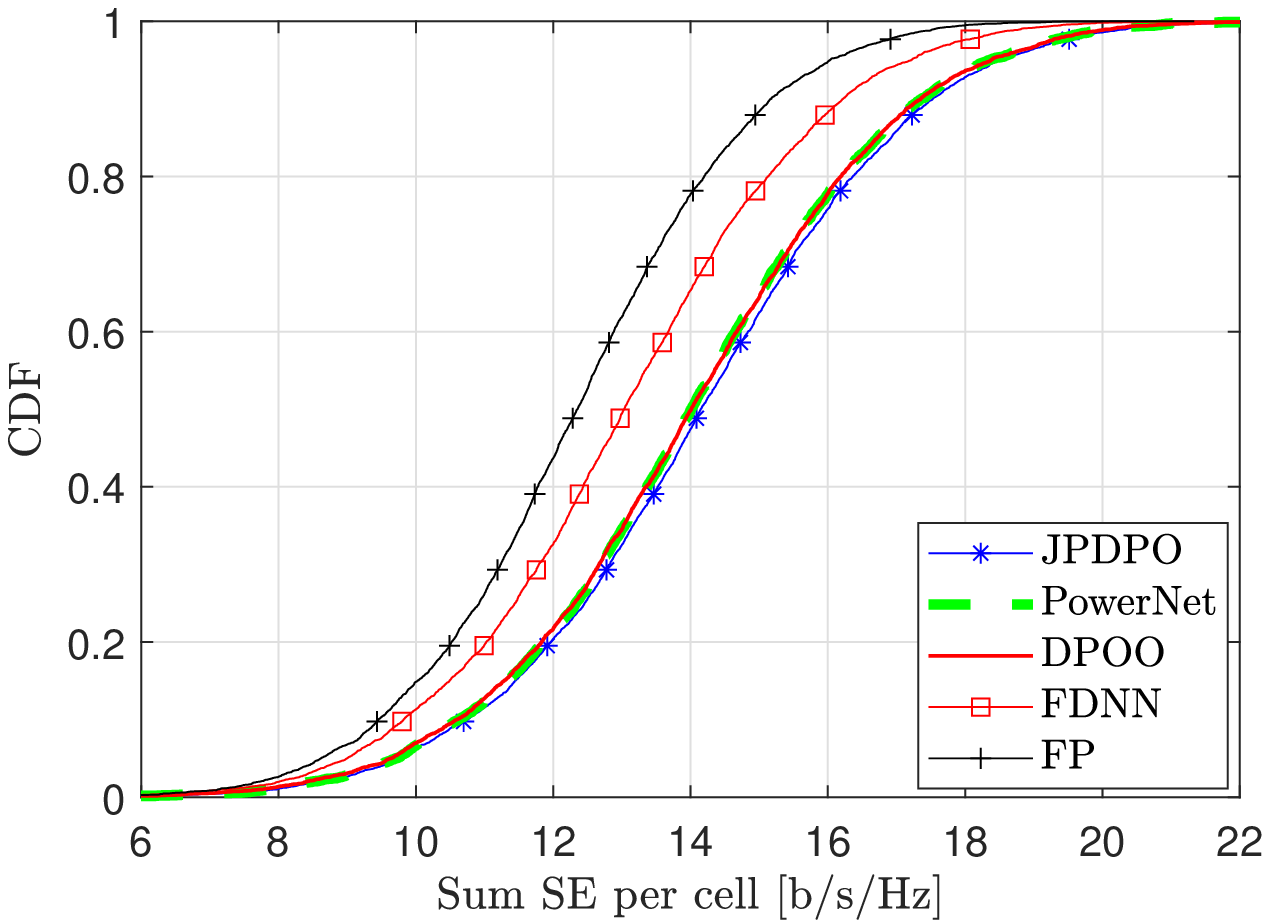} \vspace*{-0.1cm}
		\caption{CDF of SE per cell [b/s/Hz] with $L=4,K_{\max} = 10,$ and $M=200$.  All users are in active mode with a probability $2/3$ in the training phase and $1/3$ in the testing phase.}
		\label{FigSumSEperCellL4K10Ber1p3}
		\vspace*{-0.2cm}
		\noindent \end{minipage}
	\hfill
	\begin{minipage}{0.49\textwidth}
		\centering
		\includegraphics[trim=0.5cm 0.0cm 0.5cm 0.7cm, clip=true, width=3.2in]{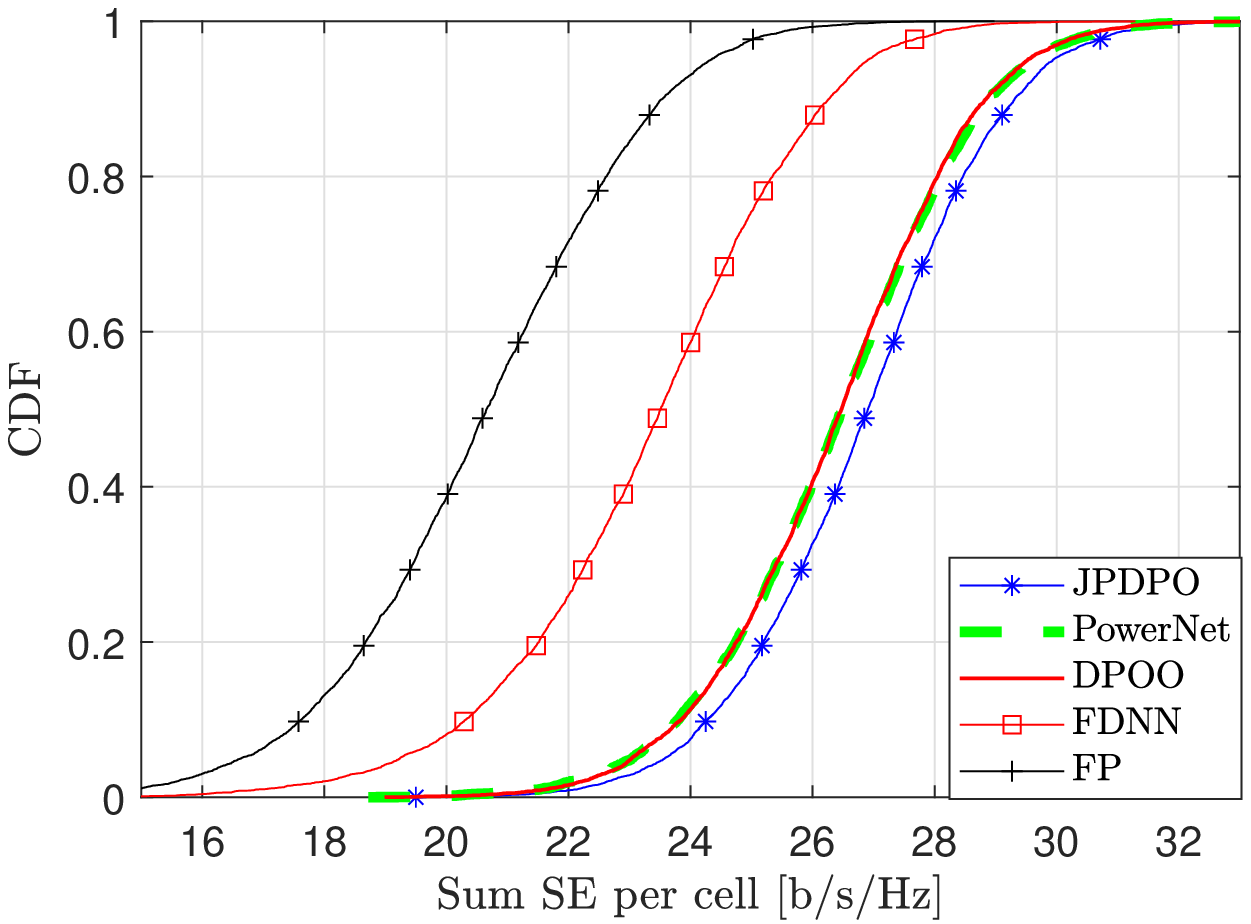} \vspace*{-0.1cm}
		\caption{CDF of SE per cell [b/s/Hz] with $L=4, K_{\max} = 10,$ and $M=200$.  All users are in active mode with a probability $2/3$ in the training phase and $5/6$ in the testing phase.}
		\label{FigSumSEperCellL4K10Ber5p6}
		\vspace*{-0.2cm}
	\end{minipage}
\end{figure*}
\begin{figure*}[t]
	\begin{minipage}{0.49\textwidth}
		\centering
		\includegraphics[trim=0.5cm 0.0cm 0.5cm 0.7cm, clip=true, width=3.2in]{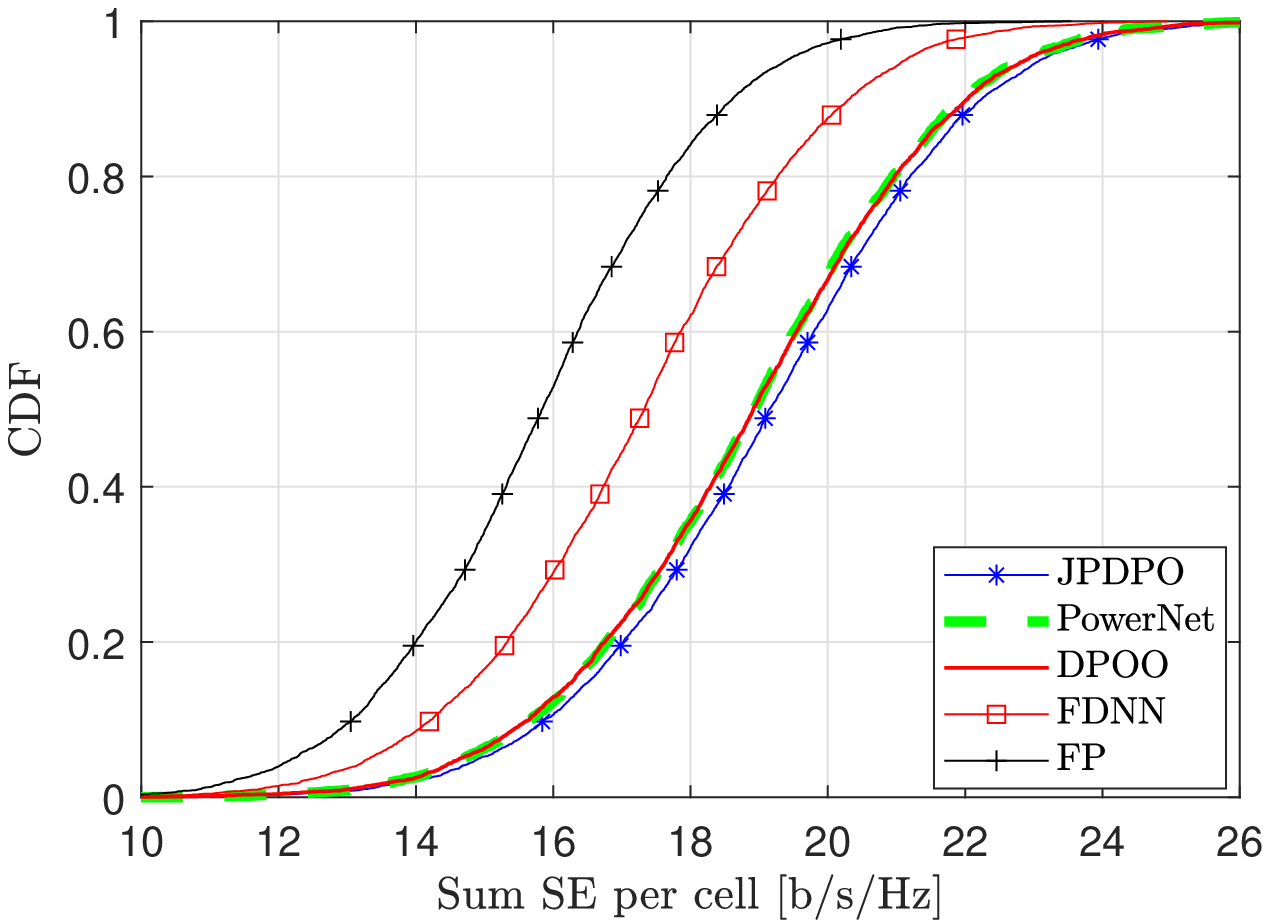} \vspace*{-0.1cm}
		\caption{CDF of SE per cell [b/s/Hz] with $L=4 ,K_{\max} = 10$, and $M=200$. Each user has the activity probability uniformly distributed in $[0,1]$.}
		\label{FigSumSEperCellL4K10BerVaries}
		\vspace*{-0.2cm}
	\end{minipage}
	\hfill
	\noindent\begin{minipage}{0.49\textwidth}
		\centering
		\includegraphics[trim=0.5cm 0.0cm 0.5cm 0.7cm, clip=true, width=3.2in]{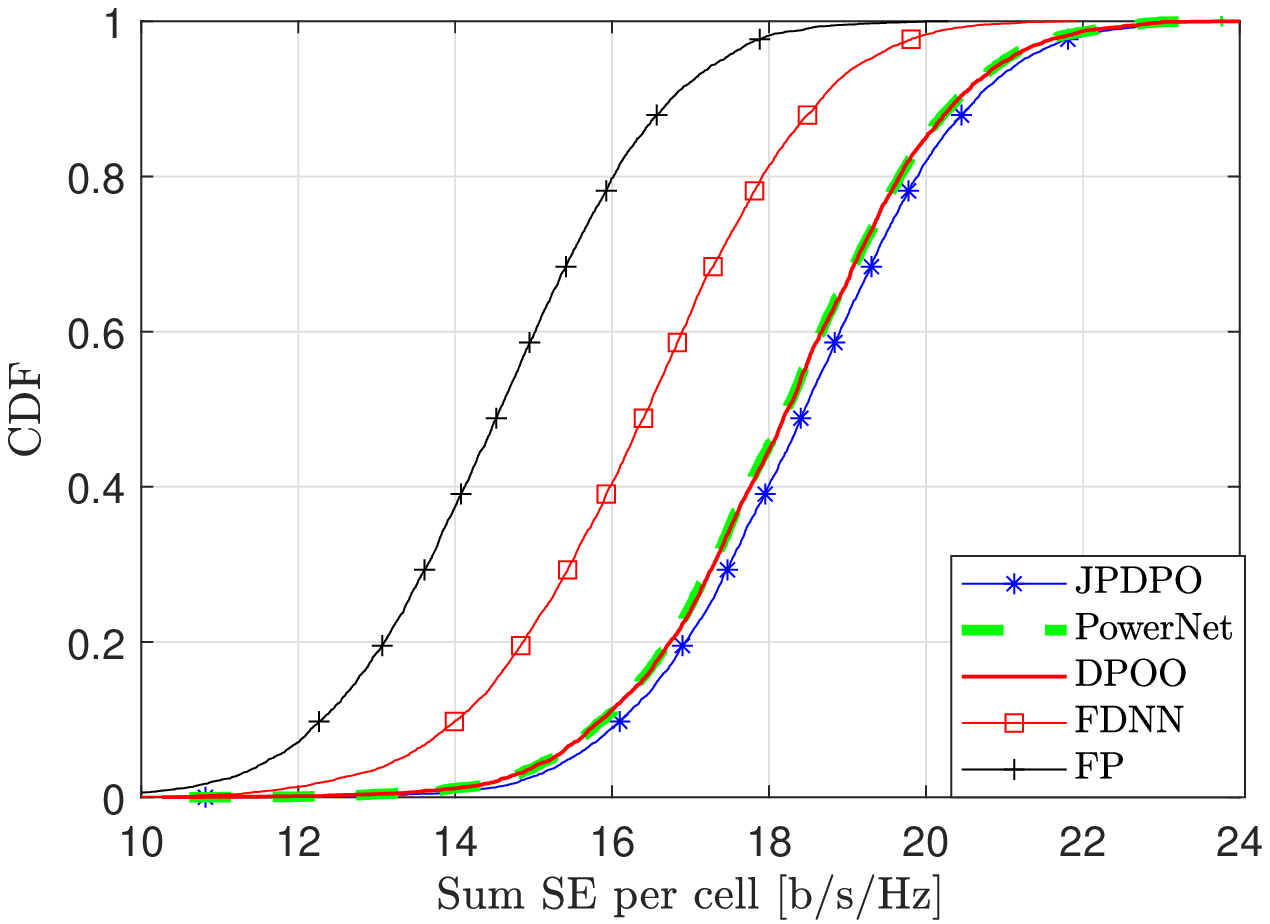} \vspace*{-0.1cm}
		\caption{CDF of SE per cell [b/s/Hz] with $L=4 ,K_{\max} = 10$, $M=200$ for the training phase, and $M=100$ for the testing phase. All users are in active mode with a probability $p=2/3$.}
		\label{FigSumSEperCellL4K10M100}
		\vspace*{-0.2cm}
		\noindent \end{minipage}
\end{figure*}
\vspace*{-0.25cm}
\subsection{Convergence, Sum Spectral Efficiency, \& Power Consumption}
Fig.~\ref{FigConvg} shows the convergence of Algorithm~\ref{Algorithm:AlternatingApproach} for different activity probabilities. Fig.~\ref{FigConvg}a is for a system with $L=4, K_{\max} =5$. It demonstrates that the number of iterations to reach the stationary point varies when the activity probability of each user changes. For example, Algorithm~\ref{Algorithm:AlternatingApproach} needs an average of $400$ iterations to converge with the activity probability $1/3$, while it only needs $300$ iterations when all users are in active mode. As another example, Fig.~\ref{FigConvg}b shows the convergence of a system with $L=9, K_{\max} =10$. Comparing the two figures, we notice that the number of iterations needed to reach the stationary point does not vary significantly when increasing the network size. At iteration~$300$, the sum SE per cell is $99\%$ of the stationary point  both with $L=4, K_{\max} =5$ and with $L=9, K_{\max} =10$.

Algorithm~\ref{Algorithm:AlternatingApproach} provides a local optimum to the sum SE optimization problem, but which local optimum that is found depends on the initialization.
One way to benchmark the quality of the obtained local optimum is to run the algorithm for many random initializations and take the best result.
Fig.~\ref{FigCDFDiffInitialization} shows the cumulative distribution function (CDF)
of the sum SE per cell obtained from Algorithm~\ref{Algorithm:AlternatingApproach} when using the best out of 1, 5, 20, or 40 different initializations.  Each initial power coefficient is uniformly distributed in the range $[0,  P_{l,k}]$. In comparison to one initialization, there are only tiny gains by spending more efforts on selecting the best out of multiple initializations. The largest relative improvement is when going from 1 to 5 initializations, but the average improvement is still less than $1\%$. Further increasing the number of initializations has very small impact on the sum SE. Hence, in the rest of this section, only one initialization is considered in the most figures, except Fig.~\ref{FigSumSEUSL}. 

We show the pilot and data power coefficients produced by our proposed methods JPDPO and DPOO in Fig.~\ref{FigPowerL4K10} for the system with $L=4, K_{\max}=10,$ and $M=200$. All users has the activity probability $2/3$. Apart from the fact that 33\% of the users are inactive on the average, an additional $5\%$  of the users are rejected from service by JPDPO due to bad channel conditions, which leads to zero power when optimizing the sum SE. By utilizing JPDPO, we observe that an user in the active mode allocates $127$ mW to each data symbol on average, while that is $150$ mW for each pilot symbol. This $18\%$ extra power is to improve the channel estimation quality. Even though many data and pilot symbols spend full power $200$ mW to achieve the best SE, JPDPO provides $25\%$ and $36\%$ less power than FP. For DPOO, we are only optimizing the data power and each data symbol is allocated $124$ mW.

\vspace*{-0.5cm}
\subsection{Predicted Performance of PowerNet}
The CDF of sum SE per cell [b/s/Hz] is shown in Fig.~\ref{FigSEperCellL4K10Ber2p3} for a system with $L=4, K_{\max}=10$, and $M=200$. Each user has the activity probability $2/3$. The sum SE per cell predicted by PowerNet is almost the same as the ones obtained DPOO and it is $1.5\%$ less than by Algorithm~\ref{Algorithm:AlternatingApproach}. FDNN performs $11.6\%$ better than FP, but there is $12.9\%$ more to reach the performance of Algorithm~\ref{Algorithm:AlternatingApproach}.  Fig.~\ref{FigPerUserSEL4K10Ber2p3} presents the prediction performance of per user SE for a four-cell system with $K_{\max}=10$ users. The SE obtained by PowerNet is very close to JPDPO with only about $1\%$ loss. Fig.~\ref{FigPerUserSEL4K10Ber2p3} also demonstrates that around $40\%$ of the users are out of service, in which case no power is allocated to the training and data transmission phases.

Fig.~\ref{FigSumSEperCellL9K10} shows the CDF of sum SE per user [b/s/Hz] for the system with $L=9, K_{\max} = 10$, and $M=200$, while the related case of per user SE is shown in Fig.~\ref{FigCDFPerUserRateL9K10}. All users have the activity probability $1$. FP provides the sum SE baseline of $23.26$ b/s/Hz which corresponds to a per user SE of $2.33$ b/s/Hz. FDNN can obtain $12.24\%$ better average SE than the baseline. In this scenario, a $4\%$ higher SE is achieved by optimizing both data and pilot powers, as compared to only optimizing the data powers. Even though the number of optimization variables is much larger than in previous figures, the average prediction error of PowerNet is still very low. The improvement of PowerNet over FP is up to $16.3\%$ for the sum SE, while it is $12.87\%$ for the per user SE. PowerNet yields $1.78\%$ better sum SE than JPDPO and the loss is only $2\%$ compared with JPDPO. These results prove the scalability of PowerNet. We emphasize that there are two main reasons why PowerNet outperforms FDNN: First, PowerNet can learn better special features from multiple observations of the large-scale fading tensors by extracting the spatial correlations among BSs based on different kernels. Second, the residual dense blocks can prevent the gradient vanishing problem effectively thanks to the extra connections between the input and output of each layer.

\vspace*{-0.25cm}
\subsection{Varying User Activity, Channel Models, and Data Label Effect}

In practice, the user activity probability will change over the day, thus it is important for PowerNet to handled this fact. Fig.~\ref{FigSumSEperCellL4K10Ber1p3} displays the CDF of the sum SE per cell [b/s/Hz] with $4$ cells, each serving $10$ users. In the training phase, each user has the activity probability $2/3$, while it is $1/3$ for the testing phase. Interestingly, PowerNet still predicts the pilot and data power coefficients very well. The sum SE per cell obtained by PowerNet is almost $99\%$ of JPDPO. Additionally, data power control is sufficient in this scenario since DPOO achieves $99\%$ of the sum SE that is produced by JPDPO. Fig.~\ref{FigSumSEperCellL4K10Ber5p6} considers a more highly-loaded system with the activity probability of each user in the testing phase being $5/6$. There is a $30\%$ gap between FP and JPDPO in this case. Furthermore, JPDPO brings $15\%$ the sum SE better than FDNN. PowerNet achieves about $98.4\%$ of what is produced by Algorithm~\ref{Algorithm:AlternatingApproach}. 

The sum SE per cell [b/s/Hz] for a system with $L=4$, $K_{\max} =10$, and $M=200$ is displayed in Fig.~\ref{FigSumSEperCellL4K10BerVaries}. In the figure, each user has its own activity probability, which is uniformly distributed in the range $[0,1]$.  FP yields the baseline average SE of $15.84$ b/s/Hz. Meanwhile, JPDPO produces the highest SE of $19.12$ b/s/Hz per cell, which is a gain of $20.71\%$. By only optimizing the data powers, DPOO loses $1.32\%$ in SE over JPDPO. Importantly, PowerNet predicts the power coefficients with high accuracy and the SE is very close to JPDPO with a loss of only $1.34\%$.
Therefore, we conclude that a single PowerNet can be trained and applied when the activity probability of users varies over time. Additionally, Fig.~\ref{FigSumSEperCellL4K10M100} displays the performance of a system where the number of antennas equipped at each BS in the testing phase is different from the training phase. It shows that PowerNet still provides very high prediction accuracy with the loss being only $1.30\%$. This indicates that PowerNet can be also applied in the scenarios where we may turn on and off antennas to improve energy-efficiency \cite{senel2019joint}. 

To provide a sensitivity analysis, Fig.~\ref{FigSumSEperCellL4K10DiffChannels} shows the prediction accuracy of PowerNet when the large-scale fading follows different distributions in the training and testing phases. In particular, the large-scale fading coefficient of user~$t$ in cell~$l$ and BS~$i$ is defined in \eqref{eq:ShadowFading} when we train PowerNet. However, the testing phase assumes $\beta_{l,t}^i \mbox{[dB]} = -151.1 - 42.8 \log_{10} (d_{l,t}^i / 1 \mathrm{km}) + z_{l,t}^i, \forall l,t,i, $ as suggested by 3GPP \cite{LTE2010b}. PowerNet still yields very high sum SE per cell and the prediction loss is less than $2\%$ compared with JPDPO.  It demonstrates that PowerNet achieves a good generalizability. We believe that this thanks to the random shadow fading realizations considered in the training phase. Meanwhile, Fig.~\ref{FigSumSEUSL} shows a zoomed-in version, where we have also trained PowerNet using the best of $40$ different initializations. At the $95\%$-likely, we observe an improvement of about $2.5\%$ over the supervised learning with using data labels from solving the optimization with one initial point. In particular, PowerNet trained using $40$ initializations outperforms JPDPO based on one initialization about $0.5\%$.

\begin{figure*}[t]
	\begin{minipage}{0.49\textwidth}
		\centering
		\includegraphics[trim=0.2cm 0.0cm 0.5cm 0.7cm, clip=true, width=3.2in]{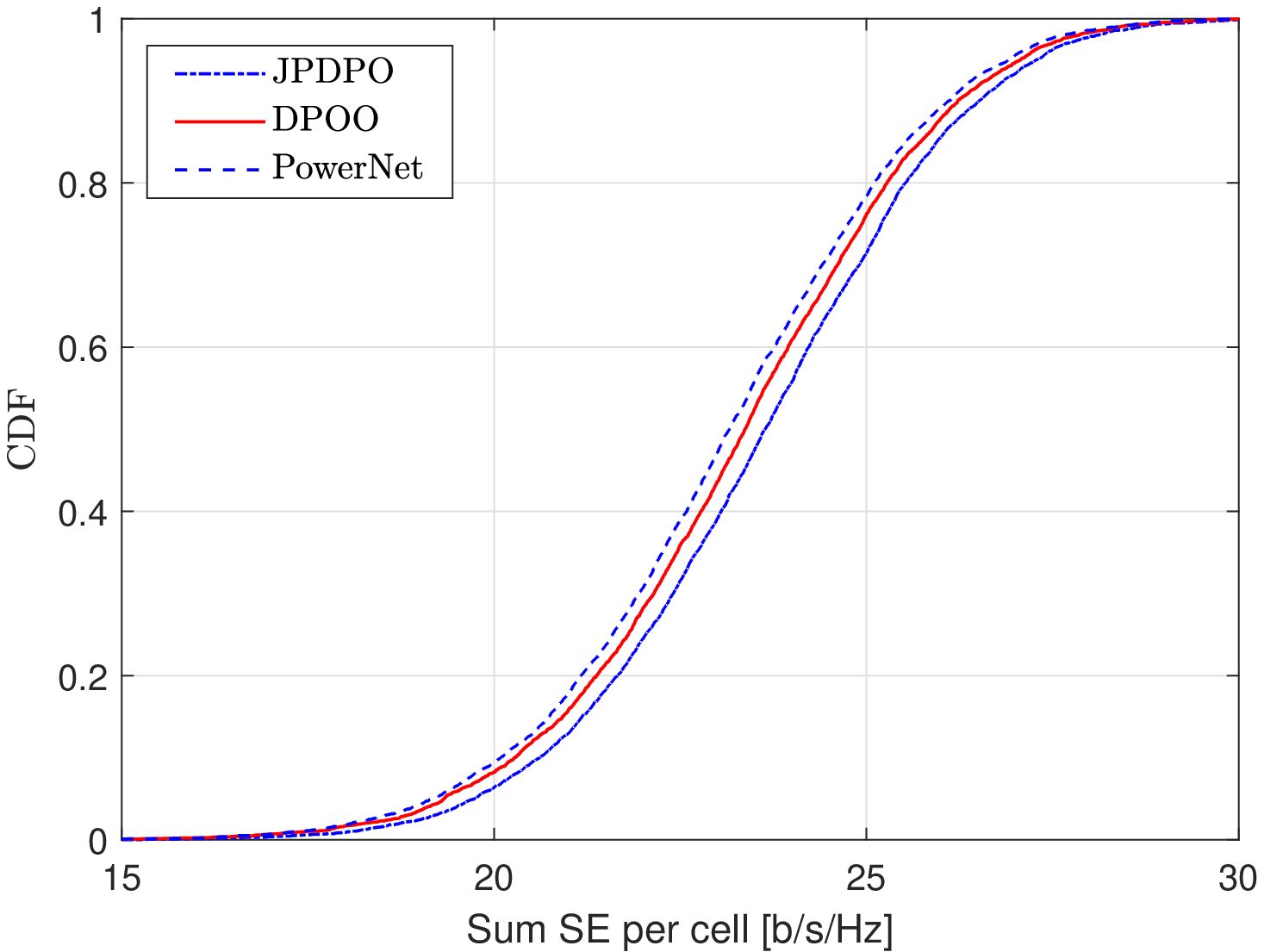} \vspace*{-0.1cm}
		\caption{Performance of PowerNet by different large-scale fading models. The system has $L=4 ,K_{\max} = 10$, and $M=200$. Activity probability is $p=2/3$.}
		\label{FigSumSEperCellL4K10DiffChannels}
		\vspace*{-0.2cm}
	\end{minipage}
	\hfill
	\noindent\begin{minipage}{0.49\textwidth}
		\centering
		\includegraphics[trim=0.2cm 0.0cm 0.5cm 0.7cm, clip=true, width=3.2in]{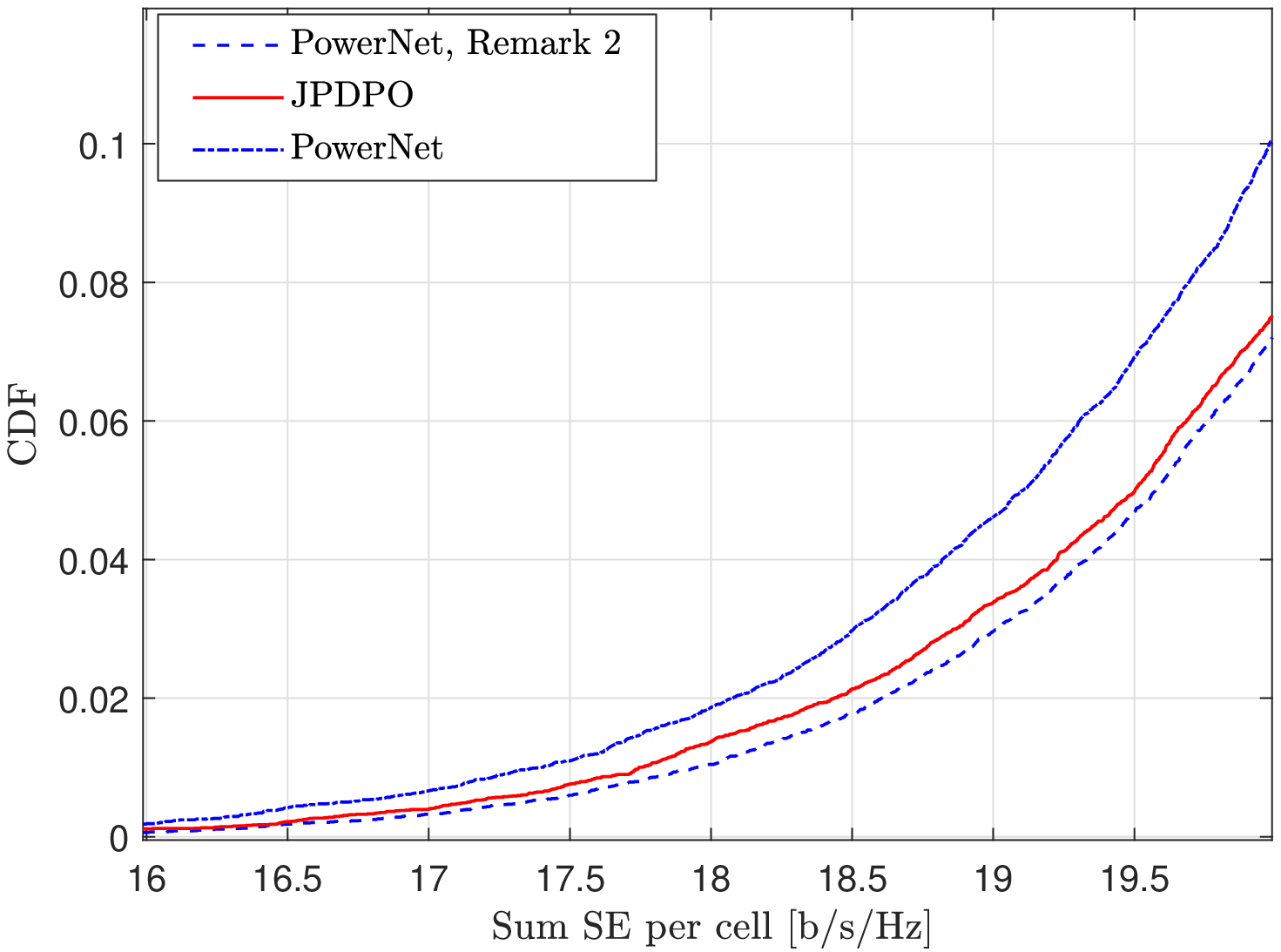} \vspace*{-0.1cm}
		\caption{Performance of PowerNet by utilizing Remark~\ref{Remark:BestLabel}. The system has $L=4 ,K_{\max} = 10$, $M=200$. All users are in active mode with a probability $p=2/3$.}
		\label{FigSumSEUSL}
		\vspace*{-0.2cm}
		\noindent \end{minipage}
\end{figure*}

\subsection{Runtime}
To evaluate the computational complexity, we implemented the testing phase in MatConvNet \cite{vedaldi2015matconvnet} and ran it on a Windows $10$ personal computer having the central processing unit (CPU) AMD Ryzen $1950$x $16$-Core with $3.40$ GHz and a Titan XP Nvidia GPU $12$~GB. Theorem~\ref{Theorem:IterativeAl} and its simplification version (DPOO) are implemented with the accuracy value $\epsilon$ in \eqref{eq:Stoping} being $0.01$ and our best efforts to optimize the implementations in Matlab environment. The average runtime is given in Table~\ref{tablerunningitme}. The algorithm in Theorem~\ref{Theorem:IterativeAl} gives the highest runtime among the considered methods.
For a system with $4$ cells, each serving maximum $5$ users, the runtime for Theorem~\ref{Theorem:IterativeAl} is $42.24$ ms, while it is $30.90$ ms if we only consider data powers as optimization variables. The proposed neural network is tested using the GPU or using only the CPU. For a system with $4$ cells, each serving $5$ users, the runtime when using the CPU is $2.99$~ms. If there are $9$ cells and $10$ users per cell, it requires approximately $5 \times$ more to obtain the solution than the system with $L=4, K_{\max} =5$. By enabling GPU mode, PowerNet can be applied for a nine-cell system ($L=9, K_{\max} =10$) with a runtime of $0.0283$~ms while taking about $0.018$~ms for a four-cell system. Hence, when using a GPU, the runtime is sufficiently low for real-time applications where each coherence interval may have a time duration of around $1$~ms and therefore require sub-millisecond resource allocation decisions.
\begin{table*}[t]
	\caption{Average running time of different methods in millisecond.} \vspace*{-0.1cm} \label{tablerunningitme}
	\centerline{ 
		\begin{tabular}{|c|c|c|c|c|}
			\hline
			\diagbox[width=14em]{System Specifications}{Benchmark} & JPDPO (CPU)  & DPOO (CPU) & PowerNet (CPU) &  PowerNet (GPU) \\
			\hline
			 $L=4,K_{\max}=5$  & $42.24$  & $30.90$ & $2.99$ & $0.0177$ \\
			\hline
			$L=9,K_{\max}=10$ & $491.08$ & $269.33$ &  $14.90$ & $0.0283$ \\
			\hline				
		\end{tabular}
	}\vspace*{-0.2cm}
\end{table*}
 
\vspace*{-0.25cm}
\section{Conclusion} \label{Section:Conclusion}
This paper has constructed a framework for the joint pilot and data power control for the sum SE maximization in uplink cellular Massive MIMO systems with a varying number of active users. This is a non-convex problem but we proposed a new iterative algorithm, inspired by the weighted MMSE approach, to find a stationary point.
 The joint pilot and data power optimization obtains 30\% higher sum SE than equal power transmission in our simulation setup. We have used the proposed algorithm to also construct a deep neural network, called PowerNet, that predicts both the data and pilot powers very well, leading to less than 1\% loss in sum SE in a symmetric multi-cell system serving 90 users.
 
 PowerNet uses only the large-scale fading coefficients to predict the transmit power, making it scalable to Massive MIMO systems with an arbitrarily large number of antennas. It has a runtime that is far below a $1$\,ms, meaning that it enables real-time power control in systems where new power control coefficients need to be obtained at the millisecond level due to changes in the scheduling decisions or user mobility. Importantly, PowerNet is designed and trained such that a single neural network can handle varying number of users per cell, which has not been the case in prior works. This demonstrates the feasibility of using deep learning for real-time power control in Massive MIMO, while still attaining basically the same performance as when solving the original problems using optimization theory. Since PowerNet can be easily implemented on standard specialized hardware that is developed for CNNs, while an efficient implementation of classical optimization algorithm requires the design of dedicated hardware circuits.
\appendix
\subsection{Proof of Theorem~\ref{Theorem:WMMSEMRC}} \label{Appendix:WMMSEMRC}
The SE of user~$k$ in cell~$l$ can be achieved by the following single-input single-output system
\begin{equation}
\tilde{y}_{l,k} = 	\sqrt{M K_{\max}} \rho_{l,k}  \hat{\rho}_{l,k} \beta_{l,k}^l x_{l,k} + w_{l,k}, 
\end{equation}
where $x_{l,k}$ is the desired real data symbol with $\mathbb{E} \{ x_{l,k}^2 \} = 1$. $w_{l,k}$ is Gaussian noise distributed as $\mathcal{N}(0, D_{l,k})$ with noting that $\rho_{l,k} = \sqrt{p_{l,k}},$  $\hat{\rho}_{l,k} = \sqrt{\hat{p}_{l,k}}, \forall l,k$, and $\mathcal{N}(\cdot, \cdot)$ being a Gaussian distribution. By using a beamforming coefficient $u_{l,k} \in \mathbb{R}$ to detect the desired signal as
\begin{equation} \label{eq:xhatlk}
\hat{x}_{l,k} = u_{l,k} \tilde{y}_{l,k} = \sqrt{M K_{\max}} \rho_{l,k}  \hat{\rho}_{l,k} \beta_{l,k}^l u_{l,k}  x_{l,k} + u_{l,k} w_{l,k}.
\end{equation}
The MSE of this decoding process is computed as
\begin{equation} \label{eq:elk}
e_{l,k} = \mathbb{E} \{ (x_{l,k} - \hat{x}_{l,k})^2 \}.
\end{equation}
Plugging the value $\hat{x}_{l,k} $ in \eqref{eq:xhatlk} into \eqref{eq:elk} and doing some algebra, we obtain the expression of $e_{l,k}$ as in \eqref{Prob:WMMSEv2}. For a given set $\{ u_{l,k}, \hat{p}_{l,k}, p_{l,k} \}$, the optimal solution to $u_{l,k}$ is obtained by taking the first-order derivative of $e_{l,k}$ with respect to $u_{l,k}$ and setting it to zero as
\begin{equation}
\begin{split}
& M K_{\max} u_{l,k} \sum_{i \in \mathcal{P}_k } (\rho_{i,k})^2 (\hat{\rho}_{i,k})^2 (\beta_{i,k}^l)^2 -  \sqrt{MK_{\max}} \rho_{l,k} \hat{\rho}_{l,k} \beta_{l,k}^l + \\
&u_{l,k} \left( K_{\max} \sum_{i \in \mathcal{P}_k} (\hat{\rho}_{i,k})^2 \beta_{i,k}^l + \sigma_{\mathrm{UL}}^2 \right) \left( \sum_{i=1}^L \sum_{t \in \mathcal{A}_i} (\rho_{i,t})^2 \beta_{i,t}^l + \sigma_{\mathrm{UL}}^2 \right)\\
&
 =0.
\end{split}
\end{equation}
Therefore the optimal solution $u_{l,k}^{\mathrm{opt}}$ is computed as in \eqref{eq:ulkopt}.
\begin{figure*}
	\begin{equation} \label{eq:ulkopt}
	\fontsize{11}{11}{\begin{split}
	&u_{l,k}^{\mathrm{opt}} = \frac{\sqrt{MK_{\max}} \rho_{l,k} \hat{\rho}_{l,k} \beta_{l,k}^l}{ MK_{\max} \sum_{i=1}^L \rho_{i,k}^2 \hat{\rho}_{i,k}^2 (\beta_{i,k}^l)^2 + \left( K_{\max} \sum_{i \in \mathcal{P}_k} \hat{\rho}_{i,k}^2 \beta_{i,k}^l + \sigma_{\mathrm{UL}}^2 \right) \left( \sum_{i=1}^L \sum_{t \in \mathcal{A}_i} (\rho_{i,t})^2 \beta_{i,t}^l + \sigma_{\mathrm{UL}}^2 \right)}.
	\end{split}}
	\end{equation}
	\hrule
	\vspace*{-0.5cm}
\end{figure*}
The optimal value to $w_{l,k}$  is computed by taking the first-order derivative of the objective function in problem~\eqref{Prob:WMMSEv1} with respect to $w_{l,k}$, and then equating it to zero:
\begin{equation} \label{eq:wlkopt}
w_{l,k}^{\mathrm{opt}} = e_{l,k}^{-1}.
\end{equation}
Using \eqref{eq:ulkopt} and \eqref{eq:wlkopt} into \eqref{Prob:WMMSEv1}, we obtain the following optimization problem
\begin{equation} \label{Prob:SumRatev1}
\begin{aligned}
& \underset{\{ \hat{p}_{l,k}, p_{l,k} \geq 0 \} }{\mathrm{minimize}}
&&   \sum_{l=1}^L |\mathcal{A}_l| - \sum_{l=1}^{L} \sum_{k \in \mathcal{A}_l } \ln \left(1  + \mathrm{SINR}_{l,k} \right)\\
& \mbox{subject to}
&&   \hat{p}_{l,k} \leq P_{\mathrm{max}, l,k }, \; \forall l,k,\\
&&& p_{l,k} \leq  P_{\mathrm{max}, l,k }, \; \forall l,k,
\end{aligned}
\end{equation}
which is easily converted to \eqref{Prob:SumRate}, so the proof is completed.

\subsection{Proof of Theorem~\ref{Theorem:IterativeAl}}  \label{Appendix:IterativeAlgMRC}
For sake of simplicity, we omit iteration index in the proof. The optimal solution to $u_{l,k}$ and $w_{l,k}$ when the other optimization variables are fixed is respectively given in \eqref{eq:ulkopt} and \eqref{eq:wlkopt} with noting that $\hat{\rho}_{l,k} = \sqrt{\hat{p}_{l,k}}$ and $\rho_{l,k} = \sqrt{p_{l,k}}$, $\forall l,k$. The feasible set to problem \eqref{Prob:WMMSEv1} is limited in the nonnegative real field, i.e.,  $\hat{\rho}_{l,k} \geq 0$ and $\rho_{l,k} \geq 0, \forall l,k$, thus the partial Lagrangian function of problem~\eqref{Prob:WMMSEv1} is
\begin{equation} \label{eq:Lagrangian}
\begin{split}
\mathcal{L} =& \sum\limits_{l=1}^L \sum\limits_{k \in \mathcal{A}_l } (w_{l,k} e_{l,k} - \ln w_{l,k}) + \sum\limits_{l=1}^L \sum\limits_{k \in \mathcal{A}_l } \lambda_{l,k} \times \\
&
\left(\rho_{l,k}^2 - P_{\max,l,k} \right)  + \sum\limits_{l=1}^L \sum\limits_{k \in \mathcal{A}_l } \mu_{l,k} \left(\hat{\rho}_{l,k}^2 - P_{\max,l,k} \right),
\end{split}
\end{equation}
where $\lambda_{l,k}$ and $\mu_{l,k}, \forall l,k,$ are Lagrange multipliers. In order to find the optimal solution to $\hat{\rho}_{l,k}$ for a given set $\{ u_{l,k}, w_{l,k}, \rho_{l,k} \}$, we take the first-order derivative of the partial Lagrangian function with respect to this variable and equaling it to zero as
\begin{equation} \label{eq:SolLambdalk1}
\begin{split}
&\hat{\rho}_{l,k} \rho_{l,k}^2 M K_{\max} \sum\limits_{i=1}^L w_{i,k} u_{i,k}^2 (\beta_{l,k}^i)^2  + \hat{\rho}_{l,k} K_{\max} \sum\limits_{j=1}^L w_{j,k}  \times \\
&
u_{j,k}^2  \beta_{l,k}^j \left( \sum\limits_{i=1}^L \sum\limits_{t \in \mathcal{A}_i} \rho_{i,t}^2 \beta_{i,t}^j + \sigma_{\mathrm{UL}}^2 \right)  - \sqrt{ M K_{\max}} \rho_{l,k} u_{l,k} w_{l,k}\beta_{l,k}^l \\
&+ \lambda \hat{\rho}_{l,k}  = 0.
\end{split}
\end{equation}
Moreover, the relationship between Lagrange multiplier $\lambda_{l,k}$ and related variable $\hat{\rho}_{l,k}$ is represented by the complementary slackness condition \cite{Boyd2004a}
\begin{equation} \label{eq:SolLambdalk2}
\lambda_{l,k} \left( \hat{\rho}_{l,k}^2 - P_{\max,l,k} \right) =0.
\end{equation}
Solving \eqref{eq:SolLambdalk1} and \eqref{eq:SolLambdalk2} gives us the optimal solution to $\hat{\rho}_{l,k}$ as in \eqref{eq:dynamicrhohatlk}. The global optimum to $\rho_{l,k}$ for a given set of $\{u_{l,k}, w_{l,k}, \hat{\rho}_{l,k} \}$ is obtained by a similar procedure. 

Algorithm~\ref{Algorithm:AlternatingApproach} must converge to a fixed point because the objective function of problem~\eqref{Prob:WMMSEv1} is a quadratic function constrained on one optimization variable while the other are predetermined and the convex feasible set ensures a bounded objective function. The objective function of problem~\eqref{Prob:WMMSEv1} is hence monotonically non-increasing over iterations \cite{Chien2018a}. We now prove that each stationary point of \eqref{Prob:WMMSEv1} is also that of problem~\eqref{Prob:SumRate}. In detail, for convenience, we first reformulate problem~\eqref{Prob:SumRate} using the natural logarithm as
\begin{equation} \label{Prob:SumRatev3}
\begin{aligned}
& \underset{\{ \hat{\rho}_{l,k}, \rho_{l,k} \geq 0 \} }{\textrm{maximize}}
&&   \sum_{l=1}^{L} \sum_{k \in \mathcal{A}_l } \ln \left(1 + \mathrm{SINR}_{l,k} \right)\\
& \textrm{subject to}
&&   \hat{\rho}_{l,k}^2 \leq P_{l,k}, \; \forall l,k,\\
&&& \rho_{l,k}^2 \leq  P_{l,k}, \; \forall l,k,
\end{aligned}
\end{equation}
then the partial Lagrangian function of problem \eqref{Prob:SumRatev3} is 
\begin{equation}
\begin{split}
\tilde{\mathcal{L}} =& \sum_{l=1}^L \sum_{k \in \mathcal{A}_l } \ln (1 + \mathrm{SINR}_{l,k}) + \sum\limits_{l=1}^L \sum\limits_{k \in \mathcal{A}_l } \lambda_{l,k} 
\left(\rho_{l,k}^2 - P_{\max,l,k} \right) \\
&+ \sum\limits_{l=1}^L \sum\limits_{k \in \mathcal{A}_l } \mu_{l,k} \left(\hat{\rho}_{l,k}^2 - P_{\max,l,k} \right).
\end{split}
\end{equation}
In order to prove the problems \eqref{Prob:SumRate} and \eqref{Prob:WMMSEv1} share the same set of stationary points, it is sufficient to prove that these equalities hold for $\forall i,t,$
$\partial \mathcal{L} / \partial \rho_{i,t} = \partial \tilde{\mathcal{L}} / \partial \rho_{i,t}$ and $\partial \mathcal{L} / \partial \hat{\rho}_{i,t} = \partial \tilde{\mathcal{L}} / \partial \hat{\rho}_{i,t}.$ 
We prove the former which is dedicated to data power control by computing $\partial \mathcal{L} / \partial \rho_{i,t}$ as
\begin{equation} \label{eq:FirstDerivativev1}
\frac{\partial \mathcal{L}}{ \partial \rho_{i,t}} = \sum_{l=1}^L \sum_{k \in \mathcal{A}_l } w_{l,k} \frac{\partial e_{l,k}}{\partial \rho_{i,t}} + 2 \lambda_{i,t} \rho_{i,t}
\end{equation}
Notice that \eqref{eq:FirstDerivativev1} holds true for $\forall w_{l,k}$ and $u_{l,k}$, thus at $w_{l,k} = w_{l,k}^{\mathrm{opt}}$ and $u_{l,k} = u_{l,k}^{\mathrm{opt}}$ we obtain
\begin{equation} \label{eq:FirstDerivativev2}
\begin{split}
 \frac{\partial \mathcal{L}}{ \partial \rho_{i,t}}&= \sum_{l=1}^L \sum_{k \in \mathcal{A}_l } (e^{\mathrm{opt}})^{-1} \frac{\partial e^{\mathrm{opt}} }{\partial \rho_{i,t}} + 2 \lambda_{i,t} \rho_{i,t}\\
 & =  \sum_{l=1}^L \sum_{k \in \mathcal{A}_l }  
 \frac{\partial \mathrm{SINR}_{l,k} }{\partial \rho_{i,t}} \left(1+ \mathrm{SINR}_{l,k} \right)^{-1} + 2 \lambda_{i,t} \rho_{i,t} \\
 &= \frac{\partial \tilde{\mathcal{L}} }{ \partial \rho_{i,t}}.
\end{split}
\end{equation}
where $w_{l,k}^{\mathrm{opt}} = 1/ e_{l,k}^{\mathrm{opt}}$ as a consequence of \eqref{eq:wlkopt}, while $e_{l,k}^{\mathrm{opt}} = \left(1+\mathrm{SINR}_{l,k} \right)^{-1}$ is gotten by plugging \eqref{eq:ulkopt} into \eqref{Prob:WMMSEv2} and doing some algebra. The procedure to get the fact that $\partial \mathcal{L} / \partial \hat{\rho}_{i,t} = \partial \tilde{\mathcal{L}} / \partial \hat{\rho}_{i,t}$ is done in the same manner.

\bibliographystyle{IEEEtran}
\bibliography{IEEEabrv,refs}
\begin{IEEEbiography} 
	{Trinh Van Chien} received the B.S. degree in Electronics and Telecommunications from Hanoi University of Science and Technology (HUST), Vietnam, in 2012. He then received the M.S. degree in Electrical and Computer Enginneering from Sungkyunkwan University (SKKU), Korea, in 2014 and the Ph.D. degree in Communication Systems from Link\"oping University (LiU), Sweden, in 2020. His interest lies in convex optimization problems for wireless communications and image \& video processing. He was an IEEE wireless communications letters exemplary reviewer for 2016 and 2017. He also received the award of scientific excellence in the first year of the 5Gwireless project funded by European Union Horizon's 2020. 
\end{IEEEbiography}
\begin{IEEEbiography} 
	{Thuong Nguyen Canh}  received a B.Sc. (2012) in Electronics and Telecommunications from Hanoi University of Science and Technology, Vietnam in 2007, M.Sc. (2014) and Ph.D. (2019) in Computer Engineering from Sungkyunkwan University, South Korea. From 2019 to 2020, He held a postdoctoral fellow at Institute for Datability Science, Osaka University, Japan. He is currently a senior researcher at Tencent America. His research interests include image/video compression standard, computational photography, wireless communication, and applications of deep learning.  
\end{IEEEbiography}
\begin{IEEEbiography} 
	{Emil Bj\"ornson} (S'07-M'12-SM'17) received the M.S. degree in engineering mathematics from Lund University, Sweden, in 2007, and the Ph.D. degree in telecommunications from the KTH Royal Institute of Technology, Sweden, in 2011. From 2012 to 2014, he held a joint post-doctoral position at the Alcatel-Lucent Chair on Flexible Radio, SUPELEC, France, and the KTH Royal Institute of Technology. He joined Link\"oping University, Sweden, in 2014, where he is currently an Associate Professor and a Docent with the Division of Communication Systems.
	
	He has authored the textbooks \emph{Optimal Resource Allocation in Coordinated Multi-Cell Systems} (2013) and \emph{Massive MIMO Networks: Spectral, Energy, and Hardware Efficiency} (2017). He is dedicated to reproducible research and has made a large amount of simulation code publicly available. He performs research on MIMO communications, radio resource allocation, machine learning for communications, and energy efficiency. Since 2017, he has been on the Editorial Board of the IEEE TRANSACTIONS ON COMMUNICATIONS and the IEEE TRANSACTIONS ON GREEN COMMUNICATIONS AND NETWORKING since 2016.
	
	He has performed MIMO research for over ten years and has filed more than twenty MIMO related patent applications. He has received the 2014 Outstanding Young Researcher Award from IEEE ComSoc EMEA, the 2015 Ingvar Carlsson Award, the 2016 Best Ph.D. Award from EURASIP, the 2018 IEEE Marconi Prize Paper Award in Wireless Communications, the 2019 EURASIP Early Career Award, the 2019 IEEE Communications Society Fred W. Ellersick Prize, and the 2019 IEEE Signal Processing Magazine Best Column Award. He also co-authored papers that received Best Paper Awards at the conferences, including WCSP 2009, the IEEE CAMSAP 2011, the IEEE WCNC 2014, the IEEE ICC 2015, WCSP 2017, and the IEEE SAM 2014.  
\end{IEEEbiography}
\begin{IEEEbiography} 
	{Erik G. Larsson} (S'99--M'03--SM'10--F'16)
	received the Ph.D. degree from Uppsala University,
	Uppsala, Sweden, in 2002.  He is currently Professor of Communication
	Systems at Link\"oping University (LiU) in Link\"oping, Sweden. He was
	with the KTH Royal Institute of Technology in Stockholm, Sweden, the
	George Washington University, USA, the University of Florida, USA, and
	Ericsson Research, Sweden.  His main professional interests are within
	the areas of wireless communications and signal processing. He 
	co-authored \emph{Space-Time Block Coding for  Wireless Communications} (Cambridge University Press, 2003) 
	and \emph{Fundamentals of Massive MIMO} (Cambridge University Press, 2016). 
	He is co-inventor of 19 issued U.S. patents.
	
	Currently he is an editorial board member of the \emph{IEEE Signal
		Processing Magazine}, and a member of the  \emph{IEEE Transactions on Wireless Communications}    steering committee. 
	He served as  chair  of the IEEE Signal Processing Society SPCOM technical committee (2015--2016), 
	chair of  the \emph{IEEE Wireless  Communications Letters} steering committee (2014--2015), 
	General respectively Technical Chair of the Asilomar SSC conference (2015, 2012), 
	technical co-chair of the IEEE Communication Theory Workshop (2019), 
	and   member of the  IEEE Signal Processing Society Awards Board (2017--2019).
	He was Associate Editor for, among others, the \emph{IEEE Transactions on Communications} (2010-2014) 
	and the \emph{IEEE Transactions on Signal Processing} (2006-2010).
	
	He received the IEEE Signal Processing Magazine Best Column Award
	twice, in 2012 and 2014, the IEEE ComSoc Stephen O. Rice Prize in
	Communications Theory in 2015, the IEEE ComSoc Leonard G. Abraham
	Prize in 2017, the IEEE ComSoc Best Tutorial Paper Award in 2018, and
	the IEEE ComSoc Fred W. Ellersick Prize in 2019.  
\end{IEEEbiography}

\end{document}